\documentclass[final]{siamart1116}
 
\usepackage{epstopdf} 
\usepackage[caption=false]{subfig}
\usepackage{amsfonts,amsmath,amsxtra}
\usepackage{bm}
\usepackage{amssymb} 
\usepackage{mathrsfs} 
\usepackage{mathtools} 
\usepackage{empheq} 
\usepackage{wrapfig,graphics,epstopdf}
\usepackage{graphicx}
\usepackage{tabularx, multirow, fancybox, lscape,listings,float,rotating,url}
\usepackage{enumerate}
\usepackage{theorem} 
\usepackage{booktabs} 
\usepackage{threeparttable} 

\usepackage[toc,page]{appendix} 

\allowdisplaybreaks[1]
\newif\ifstylew  \stylewtrue

\stylewfalse   
\newcommand{\G}{\mathbb{G}}
\newcommand{\R}{\mathbb{R}}
\newcommand{\V}{\mathcal{V}}
\newcommand{\F}{\mathcal{F}}
\newif\ifnotesw \noteswtrue

\theorembodyfont{\rmfamily}

\usepackage{cleveref}
\usepackage{arydshln}
\numberwithin{theorem}{section}
\numberwithin{equation}{section}
\numberwithin{figure}{section}
\numberwithin{table}{section}
\newcommand{\TheTitle}{Modeling the transmission of \textit{Wolbachia} in mosquitoes 
for controlling mosquito-borne diseases} 
\newcommand{\TheAuthors}{Zhuolin Qu, Ling Xue, and James M. Hyman}

\headers{Modeling \textit{Wolbachia} Transmission in Mosquitoes}{\TheAuthors}

\title{{\TheTitle}\thanks{Submitted to the editors DATE.
\funding{This research was partially supported by the NSF-MPS/NIH-NIGMS award NSF-1563531 and the NIH-NIGMS Models of Infectious Disease Agent Study (MIDAS) award U01GM097661.}}}

\author{
  Zhuolin Qu\thanks{Department of Mathematics, Tulane University, New Orleans, LA USA 70118
    (\email{zqu1@tulane.edu},
    \email{mhyman@tulane.edu}).}
  \and
  Ling Xue\thanks{Department of Mathematics, University of Manitoba, Winnipeg, MB, Canada R3T 2N2(\email{Ling.Xue@umanitoba.ca}).}
  \and
  James M. Hyman\footnotemark[2]
}

\usepackage{amsopn}

\makeatletter
\renewcommand\tableofcontents{%
  \null\hfill\textbf{\Large\contentsname}\hfill \null
  \@mkboth{\MakeUppercase\contentsname}{\MakeUppercase\contentsname}%
  \@starttoc{toc}%
}

\makeatother
   
\begin{document}

\maketitle 

\begin{abstract}
We develop and analyze an ordinary differential equation model to assess the potential effectiveness of infecting mosquitoes with the \textit{Wolbachia} bacteria to control the ongoing mosquito-borne epidemics, such as dengue fever, chikungunya, and Zika. \textit{Wolbachia} is a natural parasitic microbe that stops the proliferation of the harmful viruses inside the mosquito and reduces disease transmission.  It is difficult to sustain an infection of the maternal transmitted \textit{Wolbachia} in a wild mosquito population because of the reduced fitness of the \textit{Wolbachia}-infected mosquitoes and cytoplasmic incompatibility limiting maternal transmission.  The infection will only persist if the fraction of the infected mosquitoes exceeds a minimum threshold.  Our two-sex mosquito model captures the complex transmission-cycle by accounting for heterosexual transmission, multiple pregnant states for female mosquitoes, and the aquatic-life stage.  We identify important dimensionless numbers and analyze the critical threshold condition for obtaining a sustained \textit{Wolbachia} infection in the natural population.  This threshold effect is characterized by a backward bifurcation with three coexisting equilibria of the system of differential equations: a stable disease-free equilibrium, an unstable intermediate-infection endemic equilibrium and a stable high-infection endemic equilibrium.  We perform sensitivity analysis on epidemiological and environmental parameters to determine their relative importance to \textit{Wolbachia} transmission and prevalence.   We also compare the effectiveness of different integrated mitigation strategies and observe that the most efficient approach to establish the \textit{Wolbachia} infection is to first reduce the natural mosquitoes and then release both infected males and pregnant females.  The initial reduction of natural population could be accomplished by either residual spraying or ovitraps.
\end{abstract}



\section{Introduction}
Mathematical models can be tools to help guide mitigation efforts for zoonotic mosquito-borne diseases, such as dengue fever, chikungunya, and Zika. The symptoms for dengue fever and chikungunya include high fever, muscle and joint pains \cite{gubler1998dengue,pialoux2007chikungunya}.  Although the symptoms for Zika virus in adult human infection are usually mild, non-life threatening, a Zika infection during the pregnancy can lead to microcephaly in newborns \cite{CDCZika}. There are no effective vaccines available for these mosquito-borne diseases \cite{CDCDengueVaccine,CDCChkVaccine,CDCZikaVaccine}, and the mitigation efforts focus on the primary transmission vector, the \textit{Aedes aegypti} (\textit{Ae. aegypti}) mosquito. Most mitigation strategies focus on reducing the population size, including removing the breeding sites of mosquitoes \cite{AMCA} and indoor spraying of insecticide such as DDT.  These approaches have been proved to be effective against \textit{Ae. aegypti} mosquitoes, but the associated high financial cost, logistical difficulty in rural or urban areas and the evolution of resistance prevent it from being a reliable long-term treatment of the mosquito population \cite{lumjuan2005elevated,mcgraw2013beyond}.

Some strains of \textit{Wolbachia pipientis} (referred to as \textit{Wolbachia}) can block pathogen transmission of viruses in \textit{Ae. aegypti} \cite{walker2011wmel} and a potential strategy to reduce their ability to transmitted zoonotic diseases is to infect the wild mosquitoes with \textit{Wolbachia}. \textit{Wolbachia} is an endosymbiotic bacterium that is maternally transmitted and has been naturally found in more than $60\%$ of all insect species \cite{hilgenboecker2008many}, but not in wild \textit{Ae. aegypti} mosquitoes. The infection-induced phenomenon, cytoplasmic incompatibility (CI) \cite{laven1956cytoplasmic} that leads to early deaths of embryos produced by the crossing between an infected male mosquito and a natural female mosquito, has been employed as a bio-pesticide to eliminate natural mosquito population \cite{o2012open}. However, this strategy requires repetitive releases of a large number of infected male mosquitoes in a long run to reduce the overall population size \cite{eliminatedengue}.

If a stable population of \textit{Wolbachia}-infected mosquitoes can be established, then this approach has the potential of being both cost effective and sustainable in reducing the spread of zoonotic diseases. Since \textit{Wolbachia}-infected females are not affected by the CI phenomenon, the goal is to use the resulting reproductive advantage of infected females over uninfected ones to invade wild \textit{Ae. aegypti} population. 
Some \textit{Wolbachia} strains, such as wMelPop significantly reduces the mosquito's lifespan, 
and many of the zoonotic disease-infected mosquitoes die before they can transmit the disease to humans  \cite{mcmeniman2009stable,mcmeniman2010virulent}. Unfortunately, the reduced lifespan of wMelPop-infected mosquitoes extracts a high fitness cost and prevents the infection from being self-sustaining \cite{walker2011wmel}. The wMel \textit{Wolbachia} strain that has a lower fitness cost and high maternal transmission has been transinfected to \textit{Ae. aegypti} and successfully introduced into two areas in Australia \cite{hoffmann2014stability}. 

The success of disease control using \textit{Wolbachia} requires establishing a high-level of infection within a wild mosquito population, the key obstacle of which is to overcome the loss of fitness in the infected females, including reduced lifespan (higher death rate) and decreased fecundity (lower egg laying rate). This reduced fitness of the infected mosquitoes causes a small infection level to be cleared out, that is the disease-free state is a locally stable equilibrium. However, there is a threshold condition where if a sufficient number of mosquitoes are infected, the infection can persist. Both differential equation and discrete-time mathematical models can help understand the complex interaction of  factors that define these persistence conditions \cite{farkas2010structured,field1999microbe, hughes2013modelling, keeling2003invasion, koiller2014aedes, ndii2012modelling,  ndii2015modelling, turelli1991rapid, xue2016two}. 

Most existing ordinary differential equation (ODE) compartmental models for \textit{Wolbachia} transmission assume that there is a fixed ratio of males to females.  This assumption is a good approximation for most wild mosquito populations and can be used to reduce the model to a single-sex model with a fixed male/female ratio.  Unfortunately, this assumption is violated by some of the mitigation strategies, such as releasing only infected male mosquitoes into a wild population.  After reviewing some of the existing models, we will describe our two-sex model that also accounts for the \textit{Wolbachia} infection and the pregnancy status of the mosquitoes. 

In \cite{keeling2003invasion}, fixed sex ratio ODE models were proposed to study the competition and coexistence between multiple strains of  \textit{Wolbachia} in a well-mixed population. This paper also discussed models with spatial terms that described discretized habitats and continuous/stochastic individuals. In \cite{farkas2010structured}, a fixed male/female ratio age-structured model was proposed to incorporate different fertility and mortality rates at different stages of the life cycle of individuals, and the fitness cost was treated as increased mortality or reduced birth rate. In \cite{ndii2012modelling}, an ODE system that consisted of four compartments was used to investigate the competition between \textit{Wolbachia}-infected mosquitoes and wild mosquitoes. The authors assumed fixed ratio between male and female mosquitoes again to simplify the system and explicitly included the aquatic stage of mosquitoes and the associated resource-competition effect. Four types of steady-states were observed, depending on the maternal transmission rate, and their stability was numerically studied. 

In \cite{hughes2013modelling}, a model for dengue transmission that consisted both hosts (human being) and vectors (mosquitoes) was developed. The mosquito population were divided into uninfected and infected population, where the birth rates were parameterized from field data using a decreasing function. Like \cite{farkas2010structured}, the CI effect was reflected as reduced fertility of uninfected eggs fertilized by infected males. In \cite{ndii2015modelling}, seasonality effects in the mosquito population were introduced through the adult mosquito death rate to describe the dynamics in regions with a strong seasonal climate (distinct wet and dry periods), and the model predicted that mosquitoes carrying the wMelPop strain are less likely to persist compared with the wMel strain due to the significant reduction in lifespan.  

With few exceptions (e.g. \cite{koiller2014aedes,xue2016two}), most of these models did not stress the differences among different life stages of mosquitoes and the variant \textit{Wolbachia}-induced fitness costs for the female and male mosquitoes. Recently, in \cite{xue2016two}, a compartmental two-sex model was proposed, where the life cycle of a mosquito was divided into compartments for adult male and female mosquitoes, and an aquatic stage that combines egg, larvae and pupae.  When the basic reproductive number is less than one, the threshold effect is characterized by a backward bifurcation with three coexisting equilibria: a stable zero-infection equilibrium, an intermediate-infection unstable endemic equilibrium, and a high-infection stable endemic equilibrium (or complete infection for perfect maternal transmission). 

A female mosquito usually mates successfully once, and oviposits its eggs in different places during its entire life \cite{koiller2014aedes, foster2002mosquitoes}. Thus, when considering a two-sex model, it is important to distinguish the nonpregnant (unmated) females from the ``pregnant'' (mated) females. In \cite{koiller2014aedes}, a two-sex compartmental model of 13 ODEs explicitly included each stage of the immature mosquito (egg, larvae, pupae), and young (unmated) and fertilized (mated) females were considered separately. The fitness cost from infection was taken into account by using reduced egg laying rate for the infected females and reduced mean lifespans for both infected females and males. Under the assumption of perfect maternal transmission, three types of equilibrium were found: a stable \textit{Wolbachia} free equilibrium, a stable completely \textit{Wolbachia}-infected equilibrium and an unstable equilibrium representing the coexistence between infected and uninfected mosquitoes. 

To better understand the dynamics for the \textit{Wolbachia} invasion in a wild mosquito population, we propose a system of 9 ODEs that includes aquatic-stage mosquitoes and multiple pregnant stages for females, and we analyze the threshold condition required to sustain endemic \textit{Wolbachia} for both perfect and imperfect maternal transmissions. Our main findings are:
\begin{itemize}
\item There are three types of equilibrium: a disease-free equilibrium with no infected mosquitoes; a complete-infection equilibrium where all mosquitoes are infected; and an endemic equilibrium with both infected and uninfected mosquitoes coexisting.  
\item The epidemic can be characterized by three dimensionless numbers: the next generation number for the uninfected population, $\G_{0u}$, measures the number of uninfected eggs produced by one uninfected egg through one life cycle; the next generation number for the infected population, $\G_{0w}$, measures the number of infected eggs produced by one infected egg through one life cycle; and the basic reproductive number $\R_0 = \G_{0w}/\G_{0u}$ measures the average number of secondary infections a single \textit{Wolbachia}-infected mosquito will cause when introduced into a fully susceptible population. 
\item The backward bifurcation analysis of the proposed model indicates that when the basic reproductive number $\R_0<1$, there can still exist a stable endemic equilibrium and there is a threshold condition for the fraction of the mosquitoes that must be exceeded for a sustained \textit{Wolbachia} infection in a wild mosquito population.
\item The threshold condition can be analyzed in terms of the basic reproductive number, which is a combination of maternal transmission rate, the ratio of lifespans of infected and uninfected females, the ratio of egg laying rates for infected and uninfected females and the mating rate between a male mosquito and a nonpregnant female mosquito.
\item The best mosquito management to establish a sustained \textit{Wolbachia} infection includes using pre-release mitigation to reduce the population of wild uninfected mosquitoes before releasing a large number of \textit{Wolbachia}-infected males and pregnant females.
\end{itemize}

After describing the proposed multi-stage \textit{Wolbachia} model, we derive three types of equilibrium and their conditions of existence (\cref{sec:sec_3}), analyze the stability of the equilibria (\cref{sec:sec_4}), and characterize the threshold condition as backward bifurcation for the stable fixed points (\cref{sec:sec_5}). We then simulate and compare practical mitigation strategies in the field context (\cref{sec:sec_7}), and sensitivity analysis is performed to illustrate the key factors to the threshold condition (\cref{sec:sec_6}). 

\section{Mathematical Model}
Our multi-stage compartmental ODE model (\cref{fig:cycle}) accounts for the heterosexual transmission of \textit{Wolbachia} and the maternal transmission from infected females to their offspring. The life cycle of a mosquito is  divided into two main stages:  the aquatic stage that includes the egg, larva and pupa life stages, and the adult mosquito stage. The uninfected and the infected classes of the aquatic-stage mosquitoes are denoted by $A_u$ and $A_w$, respectively. 
The complexity induced by CI effect within maternal transmission is captured by grouping the adult mosquito population into seven compartments. The male mosquitoes can be uninfected, $M_u$,  or infected ones, $M_w$, while the nonpregnant female mosquitoes (unmated) can be uninfected, $F_u$, or infected with \textit{Wolbachia}, $F_w$.   The pregnant (mated) females can be in one of three states:  uninfected and fertile,  $F_{pu}$;   infected and sterile (the eggs laid by which don't hatch due to CI), $F_{ps}$; or infected and fertile, $F_{pw}$,  where a high percentage of their eggs are infected. 

Unlike the male mosquitoes, which could mate several times before their supplies of mature sperms and accessory gland secretion become depleted, the female mosquitoes typically mate only once and stores the sperm for several clutches of eggs.  A female rarely mates with  more than one male \cite{foster2002mosquitoes}. Our model includes separate stages for nonpregnant and pregnant female mosquitoes, and assumes there are no contacts between  male  and  pregnant female mosquitoes. 

\begin{figure}[t]
\centering
\includegraphics[width=0.9\textwidth]{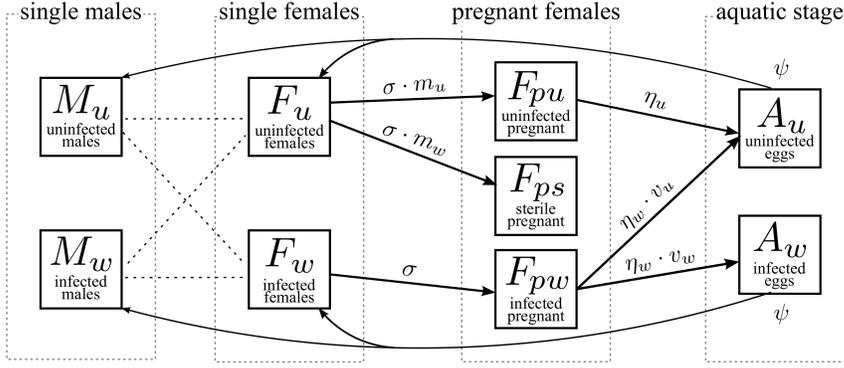}\label{fig:cycle}
\caption{Maternal transmission of \textit{Wolbachia} in the mosquito population. Uninfected females, $F_u$, and infected females, $F_w$, have contacts with either uninfected males, $M_u$, or infected males, $M_w$, once in their lives and enter one of the pregnant stages (with mating rate $\sigma$): uninfected pregnant $F_{pu}$ ($F_u$ cross $M_u$), pregnant but sterile $F_{ps}$ (CI effect: $F_u$ cross $M_w$) and infected pregnant $F_{pw}$ ($F_w$ cross either $M_u$ or $M_w$), depending on the infection status of the partners. Pregnant females start the gonotrophic cycle and produce aquatic-stage mosquitoes: uninfected pregnant females, $F_{pu}$, only produce uninfected individuals, $A_u$, (at rate $\phi_u$), pregnant sterile females, $F_{ps}$, do not have any offspring, and infected pregnant females $F_{pw}$ produce a fraction of $v_w$ infected offspring $A_w$ and a fraction of $ v_u$ uninfected offspring (at rate $\phi_w$). The aquatic-stage mosquitoes hatch and emerge into adult forms (at rate $\psi$), fraction $b_m$ of which are males and fraction $b_f$ are females.
]}
\end{figure}

We denote the per capita mortality rates of the aquatic-stage mosquitoes, the uninfected females, the infected females, the uninfected males and infected males by $\mu_a$, $\mu_{fu}$, $\mu_{fw}$, $\mu_{mu}$, and $\mu_{mw}$, respectively. We have assumed the environmental parameters remain stable, that is the changes in temperature and humidity are relatively small, so that the mortality rates are constants. We also use the same mortality rate for the infected and uninfected aquatic-stage mosquitoes, since the corresponding survival rates are not significantly different from each other \cite{walker2011wmel,mcmeniman2010virulent}. 

When there are abundant breeding sites, the egg laying rates of the uninfected females, $F_{pu}$ is $\phi_u$, and is $\phi_w$ for the infected females, $F_{pw}$. This rate is reduced by a carrying capacity, $K_a$, of the aquatic local environment, which is dependent on the availability of the breeding sites and essential environmental resources. 
Our model combines these two effects and defines the per capita egg-laying oviposition rate for uninfected and \textit{Wolbachia}-infected pregnant females as
\begin{align} 
\eta_u(A_u,A_w) =\phi_u\Big(1-\frac{A_u+A_w}{K_a}\Big)  ~,~ \text{and}~~~~
\eta_w(A_u,A_w)  = \phi_w\Big(1-\frac{A_u+A_w}{K_a}\Big) ~~.
\label{E:eta}
\end{align}

The maternal transmission efficiency, $v_w$ ($0\le v_w\le 1$), is the fraction of the offspring of  \textit{Wolbachia}-infected females that are infected and is a key parameter for establishing a sustainable population of  \textit{Wolbachia}-infected mosquitoes.  That is,  an infected pregnant female, $F_{pw}$ , lays infected eggs at the rate $v_w\eta_w$ and uninfected eggs at the rate 
$ v_u\eta_w$, where  $ v_u=1- v_w$. 
There is almost perfect maternal transmission, $ v_w \approx 1$, for the \textit{Wolbachia} strains  we are considering \cite{mcmeniman2009stable,hoffmann2014stability}.  
The aquatic-stage mosquitoes develop to adult forms at a per capita rate $\psi$, a fraction $b_f$ of which are females and $b_m=1-b_f$ are males. Typically, $b_f \approx b_m\approx 0.5$.  We assume development rate is the same in the uninfected and infected aquatic-stage population \cite{walker2011wmel,mcmeniman2010virulent}. 

The rate that nonpregnant females, $F_u$, progress to the pregnant uninfected females, $F_{pu}$, depends on the rate that nonpregnant females mate with uninfected males. We assume a constant mating rate $\sigma$ for different crosses between infected/uninfected females and infected/uninfected males. Unlike some other control strategies such as sterile insect technique \cite{alphey2010sterile} that may affect the competitiveness of the male mosquitoes, \textit{Wolbachia}-infected males are equally successful in finding and mating with females \cite{segoli2014effect}.  
When a nonpregnant female mates with a randomly selected male, the probability that the male will be uninfected is $m_u=M_u/(M_u+M_w)$.  Therefore, the $F_u$ population advances to $F_{pu}$ population at the rate $\sigma m_u$.  The rates the females advance to the other pregnant states depends on the probability that a sexual contact will be with an infected male, $m_w=1-m_u=M_w/(M_u+M_w)$, and 
can be obtained in a similar approach. 

According to the assumptions above, a model that describes the population dynamics of \textit{Wolbachia} transmission within mosquitoes is given by the following ODE system \cref{eq:ODEa,eq:ODEb,eq:ODEc,eq:ODEd,eq:ODEe,eq:ODEf,eq:ODEg,eq:ODEh,eq:ODEi}: 
\begin{subequations}
\begin{align}
\frac{d A_{u}}{d  t} &=\eta_u F_{pu}+ \eta_wv_uF_{pw}-(\mu_a+\psi)A_u~~,\label{eq:ODEa}\\
\frac{d  A_{w}}{d  t} &= \eta_wv_wF_{pw}-(\mu_a+\psi)A_w~~,\label{eq:ODEb}\\
\frac{d  F_{u}}{d  t} &= b_f \psi A_u -(\sigma+\mu_{fu}) F_u~~,\label{eq:ODEc}\\
\frac{d  F_{w}}{d  t} &= b_f \psi A_w -(\sigma+\mu_{fw}) F_w~~,\label{eq:ODEd}\\
\frac{d  F_{pu}}{d  t} &=\sigma m_uF_u-\mu_{fu} F_{pu}~~,\label{eq:ODEe}\\
\frac{d  F_{pw}}{d  t} &=\sigma  F_w-\mu_{fw} F_{pw}~~,\label{eq:ODEf}\\
\frac{d  M_{u}}{d  t} &= b_m\psi A_u -\mu_{mu} M_u~~,\label{eq:ODEg}\\
\frac{d  M_{w}}{d  t} &= b_m\psi A_w-\mu_{mw} M_w~~,\label{eq:ODEh}\\
\frac{d  F_{ps}}{d  t} &=\sigma m_w F_u-\mu_{fw} F_{ps}~~.\label{eq:ODEi}
\end{align}
\end{subequations}
The last equation \cref{eq:ODEi} for the pregnant sterile females is decoupled from the other equations and need not be considered in the stability analysis for the equilibrium states. A table of the parameter values are listed in \cref{tab:parameter}.
\begin{table}[t]
\caption{The parameters used for the \textit{Wolbachia} model. Parameter values and ranges listed below are for \textit{Ae. aegypti} mosquitoes with or without wMel strain \textit{Wolbachia} infection. The baseline values represent our best-guess estimates of the parameters in a realistic environment and are used in all the simulations, unless stated otherwise.  The Greek letter parameters are all rates with dimension $days^{-1}$. The basic reproductive number for the baseline parameters is $\R_0=0.722$ }
\label{tab:parameter}
\centering
\begin{tabular}{lllll}
\toprule
 & Description & Baseline & Range & References\\
\midrule
$b_f$ & Female birth probability & 0.5 & 0.50 -- 0.57 & \cite{tun2000effects}\\
$b_m$ & Male birth probability $=1-b_f$ & 0.5 &0.43 -- 0.50 &\cite{tun2000effects} \\
$\sigma$ &Per capita mating rate & 1 & - & Assumption \\
$\phi_u$ &Per capita egg $F_{pu}$ laying rate  & 13 & 12 -- 18 & \cite{hoffmann2014stability,mcmeniman2009stable,mcmeniman2010virulent}\\
$\phi_w$ &Per capita egg $F_{pw}$ laying rate  & 11 & 8 -- 12 & \cite{hoffmann2014stability,walker2011wmel}\\
$v_w$ &Maternal transmission efficiency & 0.95 & 0.89 -- 1 &\cite{walker2011wmel}\\
$v_u$ &$=1-v_w$ & 0.05 & 0.0 -- 0.11 & \cite{walker2011wmel}\\
$\psi$ &Per capita development rate & 1/8.75 & 1/9.2 -- 1/8.1 & \cite{hoffmann2014stability,walker2011wmel}\\
$\mu_a$& Death rate for $A_u$ or $A_w$ & 0.02 & 0.01 -- 0.04 &\cite{hoffmann2014stability,mcmeniman2010virulent,walker2011wmel}\\
$\mu_{fu}$ &Death rate for $F_u$ & 1/17.5 & 1/21 -- 1/14 &\cite{mcmeniman2009stable,styer2007mortality}\\
$\mu_{fw}$ &Death rate for $F_w$ & 1/15.8 & 1/19 -- 1/12.6 & \cite{walker2011wmel} \\
$\mu_{mu}$ &Death rate for $M_u$ & 1/10.5 & 1/14 -- 1/7 & \cite{mcmeniman2009stable,styer2007mortality} \\
$\mu_{mw}$ &Death rate for $M_w$ & 1/10.5 & 1/14 -- 1/7 & \cite{mcmeniman2009stable,styer2007mortality} \\
$K_a$ &Carrying capacity of $A_u$ or $A_w$ & $2\times 10^5$ & - &  Assume\\
\bottomrule
\end{tabular}
\end{table}

The system \cref{eq:ODEa,eq:ODEb,eq:ODEc,eq:ODEd,eq:ODEe,eq:ODEf,eq:ODEg,eq:ODEh,eq:ODEi} is epidemiologically and mathematically well-posed in the epidemiologically valid domain 
\begin{equation*}
\mathcal{D}=\left\{
\left(
\begin{array}{c}
A_u\\
A_w\\
F_u\\
F_w\\
F_{pu}\\
F_{pw}\\
F_{ps}\\
M_u\\
M_w
\end{array}
\right)\in \mathcal{R}^9~
\middle|
\begin{array}{c}
A_u\ge 0,\\
A_w\ge 0,\\
0\le A_u+A_w\le K_a,\\
F_u\ge 0,\\
F_w\ge 0,\\
0\le F_u+F_w\le \frac{b_f\psi K_a}{\sigma+\mu_{fu}},\\
F_{pu}\ge 0,\\
F_{pw}\ge 0,\\
F_{ps}\ge 0,\\
0\le F_{pu}+F_{pw}+F_{ps}\le \frac{\sigma}{\sigma+\mu_{fu}}\frac{b_f\psi K_a}{\mu_{fu}},\\
M_u\ge 0,\\
M_w\ge 0,\\
0\le M_u+M_w\le \frac{b_m\psi K_a}{\mu_{mu}}
\end{array}
\right\}.
\end{equation*}

\begin{theorem}[Forward Invariance]\label{thm:th_inv}
Assuming that the initial condition lies in domain $\mathcal{D}$, the system of equations for the maternal transmission \textit{Wolbachia} model \cref{eq:ODEa,eq:ODEb,eq:ODEc,eq:ODEd,eq:ODEe,eq:ODEf,eq:ODEg,eq:ODEh,eq:ODEi} has a unique solution that remains in $\mathcal{D}$ for all time $t>0$.
\end{theorem}
\begin{proof}
The initial value problem \cref{eq:ODEa,eq:ODEb,eq:ODEc,eq:ODEd,eq:ODEe,eq:ODEf,eq:ODEg,eq:ODEh,eq:ODEi} has a unique solution since the right-hand side (RHS) of is continuous with continuous partial derivatives in domain $\mathcal{D}$.  
To prove the domain $\mathcal{D}$ is forward-invariant, we note that along the edges of $\mathcal{D}$ the time derivatives all lead the solution into the invariant domain: 
\begin{align*}
A_u=0 &\implies A_u'\ge 0 ~~\text{\cref{eq:ODEa} since } A_u+A_w \le K_a, \text{and }\eta_u, \eta_w\ge 0 \\
A_w=0 &\implies A_w'\ge 0 ~~\text{\cref{eq:ODEb}},   \\
F_u=0 &\implies F_u'\ge 0 ~~\text{\cref{eq:ODEc}},  \\
F_w=0 &\implies F_w'\ge 0  ~~\text{\cref{eq:ODEd}}, \\
F_{pu}=0 &\implies  F_{pu}' \ge 0~~\text{\cref{eq:ODEe}}, \\
F_{pw}=0 &\implies F_{pw}' \ge 0~~\text{\cref{eq:ODEf}}, \\
M_u=0 &\implies  M_u' \ge 0 ~~\text{\cref{eq:ODEg}}, \\
M_w=0 &\implies   M_w'\ge 0 ~~\text{\cref{eq:ODEh}}, \\
F_{ps}=0 &\implies  F_{ps}'\ge 0~~\text{\cref{eq:ODEi}}. 
\end{align*}
Furthermore,
\begin{align*}
A_u+A_w = &~K_a \implies A_u'+A_w'=-(\mu_a+\psi) K_a<0, \\
F_u+F_w = &~b_f\frac{\psi}{\sigma+\mu_{fu}}K_a\implies \\
& F_u'+F_w'=b_f\psi(A_u+A_w)-(\sigma+\mu_{fu})F_u-(\sigma+\mu_{fw})F_w\\
         &\phantom{F_u'+F_w'}\le b_f\psi K_a-(\sigma+\mu_{fu})(F_u+F_w)=0, \\
F_{pu}+F_{pw}+F_{ps}=&~b_f\frac{\sigma}{\sigma+\mu_{fu}}\frac{\psi}{\mu_{fu}}K_a \implies\\
 & F_{pu}'+F_{pw}'+F_{ps}'=\sigma(F_u+F_w)-\mu_{fu}F_{pu}-\mu_{fw}(F_{pw}+F_{ps})\\
&\phantom{F_{pu}'+F_{pw}'+F_{ps}'}\le \sigma(F_u+F_w)-\mu_{fu}(F_{pu}+F_{pw}+F_{ps})\\
&\phantom{F_{pu}'+F_{pw}'+F_{ps}'}\le b_f \frac{\sigma}{\sigma+\mu_{fu}} \psi K_a-b_f\frac{\sigma}{\sigma+\mu_{fu}}\mu_{fu}\frac{\psi}{\mu_{fu}} K_a=0,\\
M_u+M_w=&~b_m\frac{\psi }{\mu_{mu}}K_a \implies\\
& M_u'+M_w'=b_m\psi(A_u+A_w)-\mu_{mu}M_u-\mu_{mw}M_w \\
& \phantom{M_u'+M_w'} \le b_m\psi K_a-\mu_{mu}(M_u+M_w)=0,
\end{align*}
where we have used the fact that  \textit{Wolbachia} bacteria increases the death rates of infected  mosquitoes, $\mu_{fw}\ge\mu_{fu}$ and $\mu_{mw} \ge \mu_{mu}$.  
Therefore, none of the orbits can leave domain $\mathcal{D}$, and there exist a unique solution.
\end{proof}

\section{Equilibria and Basic Reproductive Number}\label{sec:sec_3}
There are three types of equilibrium points, corresponding to distinct disease spreading situations, that are associated with system \cref{eq:ODEa,eq:ODEb,eq:ODEc,eq:ODEd,eq:ODEe,eq:ODEf,eq:ODEg,eq:ODEh}: disease-free equilibrium (DFE), complete-infection equilibrium (CIE) and endemic equilibrium (EE).  After describing the DFE and CIE equilibrium, we derive the basic reproductive number for the model before analyzing the EE.

\subsection{Disease-free Equilibrium (DFE)}\label{sec:sec_31}
Although \textit{Wolbachia} is found in more than $60\%$ of the insect species \cite{hilgenboecker2008many}, it is not found in wild \textit{Ae. aegypti} because of the loss of fitness it causes in \textit{Ae. aegypti}. In other words, without artificially introducing \textit{Wolbachia} into the field, the wild \textit{Ae. aegypti} mosquito population will be at the DFE.

The DFE is found by setting $A_w=F_w=F_{pw}=M_w=0$, and the unique non-trivial steady-state is denoted by $EE^0=(A_u^0,0,F_u^0,0,F_{pu}^0,0,M_u^0,0)$, where
\begin{align}
A_u^0&=K_a\left(1-\frac{1}{\G_{0u}}\right),\nonumber\\
F_u^0&=b_f\frac{\psi}{\mu_{fu}+\sigma}A_u^0,\label{eq:DFE}\\
F_{pu}^0&=b_f\frac{\psi\sigma}{(\mu_{fu}+\sigma)\mu_{fu}}A_u^0,\nonumber\\
M_u^0&=b_m\frac{\psi}{\mu_{mu}}A_u^0.\nonumber
\end{align}

The next generation number for the uninfected population, 
\begin{equation}
\G_{0u} = b_f\frac{\psi}{\mu_a+\psi}\frac{\sigma}{\sigma+\mu_{fu}}\frac{\phi_u}{\mu_{fu}},
\label{eq:R0u}
\end{equation}
represents the number of uninfected eggs that one uninfected egg can generate within one life cycle of a mosquito. 
This dimensionless number can be interpreted biologically where $1/(\mu_a+\psi)$ is the average time of being in the aquatic stage, $\psi$ is the average per capita developing rate, and $b_f$ is the fraction of an aquatic-stage individual becoming a female adult.  Their product, $b_f\,\psi/(\mu_a+\psi)$, is the probability that an uninfected egg develops into a nonpregnant uninfected female (in compartment $F_{u}$). Similarly, $\sigma/(\sigma+\mu_{fu})$ is the probability that an uninfected nonpregnant female becomes a pregnant uninfected mosquito (at DFE, all males are uninfected), and $\phi_u/\mu_{fu}$ is the average number of eggs that an uninfected pregnant female can produce before it dies.

In a wild mosquito population without \textit{Wolbachia} infection, $\G_{0u}>1$ is the essential condition that guarantees the persistence of the natural population and, therefore,  we assume that $\G_{0u}>1$ .

\subsection{Complete-infection Equilibrium (CIE)}
When the maternal transmission is perfect ($v_w=1$), that is all the offspring produced by the infected pregnant females are infected, it is possible that the \textit{Wolbachia} infection can spread throughout the entire mosquito population. The CIE is found by setting $A_u=F_u=F_{pu}=M_u=0$ in the system \cref{eq:ODEa,eq:ODEb,eq:ODEc,eq:ODEd,eq:ODEe,eq:ODEf,eq:ODEg,eq:ODEh} and can only happen when $v_w=1$. 
This condition can be derived from  \cref{eq:ODEa} where the  term $v_u\eta_wF_{pw}$ has to be zero at CIE. Let  $EE^c=(0,A_w^c,0,F_w^c,0,F_{pw}^c,0,M_w^c)$ denote the CIE, where 
\begin{align}
A_w^c&=K_a\left(1-\frac{1}{\G_{0w}}\right),\nonumber\\
F_w^c&=b_f\frac{\psi}{\mu_{fw}+\sigma}A_w^c,\label{eq:CIE}\\
F_{pw}^c&=b_f\frac{\psi\sigma}{(\mu_{fw}+\sigma)\mu_{fw}}A_w^c,\nonumber\\
M_w^c&=b_m\frac{\psi}{\mu_{mw}}A_w^c.\nonumber
\end{align}
The next generation number for the infected population 
\begin{equation}
\G_{0w} = v_w b_f\frac{\psi}{\mu_a+\psi}\frac{\sigma}{\sigma+\mu_{fw}}\frac{\phi_w}{\mu_{fw}},
\label{eq:R0w}
\end{equation}
represents the number of infected eggs that one infected egg can generate within one life cycle of a mosquito.
Here $v_w=1$ in the case of perfect maternal transmission. As the dimensionless number $\G_{0u}$ introduced in \cref{eq:R0u}, $\G_{0w}$ can also be interpreted biologically as follows: as before, $b_f\,\psi/(\mu_a+\psi)$ is the probability that an infected aquatic-stage egg develops into an infected female adult, $\sigma/(\sigma+\mu_{fw})$ is the probability that an infected nonpregnant female becomes a pregnant infected one (at CIE, only infected males present for mating), and $ v_w\,\phi_w/\mu_{fw}$ is the average number of infected eggs that an infected pregnant female can produce. 

When $v_w<1$, there can still be an infected EE, but it will not be a CIE. We will characterize this EE after first defining the basic reproductive number. 

\subsection{The Basic Reproductive Number \texorpdfstring{$\R_0$}{R0}}

The basic reproductive number $\R_0$ serves as a threshold condition and determines the initial establishment of disease transmission in a totally susceptible population. We derive this dimensionless number directly from the ODE system \cref{eq:ODEa,eq:ODEb,eq:ODEc,eq:ODEd,eq:ODEe,eq:ODEf,eq:ODEg,eq:ODEh} by using the next generation method \cite{van2002reproduction}.  In the next generation analysis, we first collect all the infected compartments of the system, $\mathbf{X}=(A_w,F_w,F_{pw},M_w)^T$, which correspond to \cref{eq:ODEb}, \cref{eq:ODEd}, \cref{eq:ODEf} and \cref{eq:ODEh}, 
and split the right hand side of \cref{eq:ODEb,eq:ODEd,eq:ODEf,eq:ODEh} into two parts: the rates of new infections $\F$ and the rates of transitions $\V$:
\begin{align*}
\frac{d\mathbf{X}}{d t}=
\frac{d}{d t}
\left(
\begin{array}{c}
A_w\\
F_w\\
F_{pw}\\
M_w
\end{array}
\right)
&=
\left(
\begin{array}{c}
 v_w\eta_wF_{pw}\\
0\\
0\\
0
\end{array}
\right)
-
\left(
\begin{array}{c}
(\mu_a+\psi)A_w\\
-b_f\psi A_w+(\sigma+\mu_{fw})F_w\\
-\sigma F_w+\mu_{fw}F_{pw}\\
-b_m\psi A_w+\mu_{mw}M_w
\end{array}
\right)
=:\F-\V.
\end{align*}
The Jacobian matrices of $\F$ and $\V$ at DFE \cref{eq:DFE} are given by
\begin{align*}
J_{\F}&:=\frac{\partial \F}{\partial \mathbf{X}}
=
\left(
\begin{array}{cccc}
0 & 0& v_w\eta_w(A_u^0,0)&0\\
0&0&0&0\\
0&0&0&0\\
0&0&0&0
\end{array}
\right)~~, \text{and}\\
J_{\V}&:=\frac{\partial \V}{\partial \mathbf{X}}
=
\left(
\begin{array}{cccc}
\mu_a+\psi & 0 & 0 & 0\\
-b_f\psi   &\sigma+\mu_{fw} & 0 & 0\\
0          &-\sigma  & \mu_{fw} & 0\\
-b_m\psi  & 0 & 0 & \mu_{mw}
\end{array}
\right)~~.
\end{align*}
The basic reproductive number is calculated as the spectral radius of the next generation matrix $J_{\F}J_{\V}^{-1}$, 
\begin{equation}
\R_0:=\text{Spectral Radius of }(J_{\F}J_{\V}^{-1})=v_w\,\frac{\mu_{fu}\,\phi_w\,(\sigma+\mu_{fu})}{\mu_{fw}\,\phi_u\,(\sigma+\mu_{fw})}~~,
\label{eq:R0}
\end{equation}
and is a linear function of the vertical transmission rate, $v_w$, for \textit{Wolbachia}. 

The role of $v_w$ arises from its role in next generation number for the infected population $\G_{0w}$ and becomes clear when we write  
$\R_0$ \cref{eq:R0}  as
\begin{equation*}
\R_0=\left.\left( v_w b_f\frac{\psi}{\mu_a+\psi}\frac{\sigma}{\sigma+\mu_{fw}}\frac{\phi_w}{\mu_{fw}}\right) \middle/ \left(b_f\frac{\psi}{\mu_a+\psi}\frac{\sigma}{\sigma+\mu_{fu}}\frac{\phi_u}{\mu_{fu}}\right)\right.=\frac{\G_{0w}}{\G_{0u}}.
\end{equation*}
Recall that the biological interpretations of dimensionless numbers $\G_{0w}$ and $\G_{0u}$, and $\R_0$ can be interpreted as the factor for how much the ratio of new infected to new uninfected eggs changes from one generation to the next.

If $\R_0>1$, then a small \textit{Wolbachia} infection would eventually spread throughout the population. Unfortunately, \textit{Wolbachia} infection deceases the fitness of the infected mosquitoes, that is $\G_{0w}<\G_{0u}$ ($\R_0<1$), thus a small \textit{Wolbachia} infection introduced at the DFE will die out. For the baseline case, based on our best estimates for the model parameters, $\R_0=0.72$. 
However, this linear analysis is based on small perturbations about the DFE. When a large infection is introduced, the endemic \textit{Wolbachia} may still happen. We will use backward bifurcation analysis to describe this threshold condition. 

\subsection{Endemic Equilibrium (EE)}\label{sec:sec_34}
Both field releases \cite{hoffmann2014stability} and lab experiments \cite{walker2011wmel} have shown that maternal transmission is not perfect, that is $v_w<1$. Under this situation, CIE could not be achieved. Instead, there are endemic states, where infected and uninfected mosquitoes could coexist in the mosquito population.

The ratio of the infected and uninfected aquatic states as $r_{wu}=A_w/A_u$ is a key parameter in defining the EE.  We assume that $\mu_{mw}=\mu_{mu}$, since \textit{Wolbachia} infection does not affect the lifespan of the males significantly in general. We let $EE^*=(A_u^*,A_w^*,F_u^*,F_w^*,F_{pu}^*,F_{pw}^*,M_u^*,M_w^*)$ denote the EE, where
\begin{align*}
\begin{split}
A_u^*&=\frac{K_a}{1+r_{wu}}\left(1-\frac{1}{\G_{0w}}\right),\\
A_w^*&=r_{wu}\,A_u^*,\\
F_u^*&=b_f\,\frac{\psi}{\sigma+\mu_{fu}}A_u^*,\\
F_w^*&=r_{wu}\,b_f\,\frac{\psi}{\sigma+\mu_{fw}}A_u^*,\\
F_{pu}^*&=\frac{1}{1+r_{wu}}\,b_f\,\frac{\psi\sigma}{(\mu_{fu}+\sigma)\mu_{fu}}A_u^*,\\
F_{pw}^*&=r_{wu}\,b_f\,\frac{\psi\sigma}{(\mu_{fw}+\sigma)\mu_{fw}}A_u^*,\\
M_u^*&=b_m\,\frac{\psi}{\mu_{mu}}A_u^*,\\
M_w^*&=r_{wu}\,b_m\,\frac{\psi}{\mu_{mw}}A_u^*,
\end{split}
\end{align*}
and ratio $r_{wu}>0$ satisfies the following equation
\begin{equation}
\frac{v_u}{v_w}r_{wu}^2+\left(\frac{v_u}{v_w}-1\right)r_{wu}+\frac{1-\R_0}{\R_0}=0,
\label{eq:eqn_r}
\end{equation}
where $\R_0$ is the basic reproductive number defined in \cref{eq:R0}.

When there is perfect maternal transmission ($v_w=1$), \cref{eq:eqn_r} is linear with the solution
\begin{equation}
r_{wu}^* =\frac{A_w^*}{A_u^*}=\frac{1-\R_0}{\R_0}  \quad \mbox{when}\quad 0<\R_0<1,\label{eq:EEstar}
\end{equation} 
and we denote the corresponding unique EE as $EE^*$.

When there is imperfect maternal transmission ($v_w<1$), there are two roots for equation \cref{eq:eqn_r}
\begin{align}
r_{wu}^+ &=\frac{1}{2 v_u}\left(2 v_w-1+\sqrt{1-\frac{4 v_u v_w}{\R_0}}\right)  \quad\mbox{and}\quad \label{eq:EEplus}\\
r_{wu}^- &=\frac{1}{2 v_u}\left(2 v_w-1-\sqrt{1-\frac{4 v_u v_w}{\R_0}}\right) \label{eq:EEminus}, 
\end{align}
corresponding to two EE, denoted by $EE^+$ and $EE^-$. The roots must be real and positive for the EE to be physically meaningful. This implies that there is no EE when $\R_0<4 v_u v_w$. 

Assume $0.5<v_w<1$ (for most strains of \textit{Wolbachia}  $v_w \approx 1$), then $4 v_u v_w=4(1-v_w)v_w<1$, and we have the following:
\begin{enumerate}[(i)]
\item when $\R_0=4 v_u v_w$, there is a single root $r_{wu}^\pm=r_{wu}^+= r_{wu}^-= (2v_w-1)/(2 v_u)$ and a single $EE^\pm=EE^+=EE^-$;
\item when $4 v_u v_w<\R_0<1$, we have $r_{wu}^+>r_{wu}^->0$, and there are two meaningful EE, $EE^+$ and $EE^-$; 
\item when $\R_0\ge 1$, $r_{wu}^-\le 0$ and only the positive root $r_{wu}^+$ and $EE^+$ is physically meaningful.  
\end{enumerate}
Note that $v_w \approx 1$ for the strains we are considering and condition becomes $\R_0>4 v_u v_w\approx 0$.

\section{Stability and Bifurcation Analysis}\label{sec:sec_4}
The stability of these equilibria is governed by the sign of the eigenvalues of the Jacobian for the equations, \cref{eq:ODEa,eq:ODEb,eq:ODEc,eq:ODEd,eq:ODEe,eq:ODEf,eq:ODEg,eq:ODEh}, linearized about each equilibrium point  (\cref{tab:table_s}).  The solution dynamics can then be characterized by using bifurcation diagrams to illustrate the threshold conditions for establishing an endemic \textit{Wolbachia}-infected population.

To simplify the structure of the Jacobian of nonlinear system \cref{eq:ODEa,eq:ODEb,eq:ODEc,eq:ODEd,eq:ODEe,eq:ODEf,eq:ODEg,eq:ODEh},
we rearrange the order of compartments as $\mathbf{Y}=(A_u, F_u, F_{pu}, M_u, A_w, F_w, F_{pw}, M_w)$. The corresponding Jacobian of the rearranged system, $\frac{d\mathbf{Y}}{dt}=J\mathbf{Y}$,  is 
\begin{align}
\mathbf{J}&=\left(\begin{array}{c;{2pt/2pt}c}
A & B\\ \hdashline[2pt/2pt]
C & D
\end{array}
\right) \\ 
&=\left(
\begin{array}{cccc;{2pt/2pt}cccc}
a_{11} & 0 & \eta_u & 0 & b_{11} & 0 & v_u \eta_w & 0 \\
b_f\psi & -\sigma-\mu_{fu} & 0 & 0 & 0 & 0 & 0 & 0 \\
0 & \sigma m_u & -\mu_{fu} & a_{34} & 0 & 0 & 0 & b_{34}\\
b_m\psi & 0 & 0 & -\mu_{mu} & 0 & 0 & 0 & 0\\ \hdashline[2pt/2pt]
c_{11} & 0 & 0 & 0 & d_{11} & 0 & v_w\eta_w & 0\\
0 & 0 & 0 & 0 & b_f\psi & -\sigma-\mu_{fw} & 0 & 0\\
0 & 0 & 0 & 0 & 0 & \sigma & -\mu_{fw} & 0\\
0 & 0 & 0 & 0 & b_m\psi & 0 & 0 & -\mu_{mw}
\end{array}
\right) 
\label{eq:Jac}
\end{align}
where 
\begin{align*}
a_{11} &= -\phi_u\frac{F_{pu}}{K_a}- v_u\phi_w\frac{F_{pw}}{K_a}-(\mu_a+\psi)~,~ &a_{34}&=\sigma m_w \frac{F_u}{M_u+M_w},\\
b_{11}  &= -\phi_u\frac{F_{pu}}{K_a}- v_u\phi_w\frac{F_{pw}}{K_a}~,~ &b_{34}&=-\sigma m_u \frac{F_u}{M_u+M_w},\\
c_{11} &=- v_w\phi_w\frac{F_{pw}}{K_a}~,~ &d_{11}&=- v_w\phi_w\frac{F_{pw}}{K_a}-(\mu_a+\psi).
\end{align*}

\subsection{Stability of the Disease-free Equilbirium}
At the DFE, we write the Jacobian as \cref{eq:Jac}

\begin{equation}
J_{DFE}=\left(
\begin{array}{c;{2pt/2pt}c}
A_{DFE} & B_{DFE}\\ \hdashline[2pt/2pt]
0 & D_{DFE}
\end{array}
\right),
\label{eq:Jac_DFE}
\end{equation}
where 
\begin{equation*}
A_{DFE} = \left(
\begin{array}{cccc}
-\G_{0u}(\mu_a+\psi) & 0 & \frac{\phi_u}{\G_{0u}} & 0\\
b_f\psi & -\sigma-\mu_{fu} & 0 & 0\\
0 & \sigma & -\mu_{fu} & 0\\
b_m\psi & 0 & 0 & -\mu_{mu}
\end{array}
\right),
\end{equation*}
and 
\begin{equation*}
D_{DFE} = \left(
\begin{array}{cccc}
-(\mu_a+\psi) & 0 & \frac{v_w\phi_w}{\G_{0u}} & 0\\
b_f\psi & -\sigma-\mu_{fw} & 0 & 0\\
0 & \sigma & -\mu_{fw} & 0\\
b_m\psi & 0 & 0 & -\mu_{mw}
\end{array}
\right).
\end{equation*}

Because $J_{DEF}$ is an upper triangular block matrix, the eigenvalues of matrix $J_{DFE}$ are the collection of those for matrix $A_{DFE}$ and $D_{DFE}$.

\begin{theorem}[Stability of Disease-free Equilibrium]\label{thm:th_DFE}
The disease-free equilibrium $EE^0=(A_u^0,0,F_u^0,0,F_{pu}^0,0,M_u^0,0)$ and \cref{eq:DFE} of the system \cref{eq:ODEa,eq:ODEb,eq:ODEc,eq:ODEd,eq:ODEe,eq:ODEf,eq:ODEg,eq:ODEh} is locally asymptotically stable (LAS) if $\G_{0u}>1$ and $\R_0<1$.
\end{theorem}

\begin{proof}
To prove the stability of the matrices, we apply a result on Metzler matrices (Proposition 3.1 in \cite{kamgang2008computation}). 
At the DFE, the Jacobian is partitioned as
\begin{equation*}
J_{DFE}=\left(
\begin{array}{c;{2pt/2pt}c}
A_{DFE} & B_{DFE}\\ \hdashline[2pt/2pt]
0 & D_{DFE}
\end{array}
\right).
\end{equation*}

We first prove the stability of matrix
\begin{equation*}
A_{DFE} = \left(
\begin{array}{cccc}
-\G_{0u}(\mu_a+\psi) & 0 & \frac{\phi_u}{\G_{0u}} & 0\\
b_f\psi & -\sigma-\mu_{fu} & 0 & 0\\
0 & \sigma & -\mu_{fu} & 0\\
b_m\psi & 0 & 0 & -\mu_{mu}
\end{array}
\right).
\end{equation*}
The $(4,4)$ element of matrix $A_{DFE}$, $-\mu_{mu}<0$, is a negative eigenvalue. Therefore, we can reduce the problem to considering  the $3\times 3$ leading principal submatrix of $A_{DFE}$, which we partitioned as 
\begin{equation*}
A_{s1} =
\left(
\begin{array}{cc;{2pt/2pt}c}
-\G_{0u}(\mu_a+\psi) & 0 & \frac{\phi_u}{\G_{0u}}\\
b_f\psi & -\sigma-\mu_{fu} & 0\\ \hdashline[2pt/2pt]
0 & \sigma & -\mu_{fu} 
\end{array}
\right)
 =\left(
\begin{array}{c;{2pt/2pt}c}
A_1 & B_1\\\hdashline[2pt/2pt]
C_1 & D_1
\end{array}
\right).
\end{equation*}
$A_{s1}$ is a Metzler matrix \cite{kamgang2008computation} and is Metzler stable if and only if both $A_1$ and $D_1-C_1A_1^{-1}B_1$ are Metzler stable.  Metzler stability of $A_1$ follows because it is a lower triangular matrix with negative diagonal entries and nonnegative off-diagonal entries, and 
\begin{equation*}
D_1-C_1A_1^{-1}B_1=-\mu_{fu}\left(1-\frac{1}{\G_{0u}}\right)<0 \quad\mbox{provided}\quad \G_{0u}>1~.
\end{equation*}

Now we consider the stability of
\begin{equation*}
D_{DFE} = \left(
\begin{array}{cccc}
-(\mu_a+\psi) & 0 & \frac{v_w\phi_w}{\G_{0u}} & 0\\
b_f\psi & -\sigma-\mu_{fw} & 0 & 0\\
0 & \sigma & -\mu_{fw} & 0\\
b_m\psi & 0 & 0 & -\mu_{mw}
\end{array}
\right).
\end{equation*}
The $(4,4)$ entry: $-\mu_{mw}<0$ is a negative eigenvalue of $D_{DFE}$, and therefore we need only  consider the $3\times 3$ leading principal submatrix
\begin{equation*}
D_{s1} = \left(
\begin{array}{cc;{2pt/2pt}c}
-(\mu_a+\psi) & 0 & \frac{v_w\phi_w}{\G_{0u}} \\
b_f\psi & -\sigma-\mu_{fw} & 0 \\ \hdashline[2pt/2pt]
0 & \sigma & -\mu_{fw} 
\end{array}
\right) =\left(
\begin{array}{c;{2pt/2pt}c}
A_2 & B_2\\\hdashline[2pt/2pt]
C_2 & D_2
\end{array}
\right),
\end{equation*}
which is a Metzler matrix. Since  $A_2$ is Metzler stable and 
\begin{equation*}
D_2-C_2A_2^{-1}B_2=-\mu_{fw}(1-\R_0)<0 \quad\mbox{provided}\quad \R_0<1,
\end{equation*}
thus $D_{s1}$ is Metzler stable.

Therefore, the Jacobian $J_{DFE}$ is stable, all the eigenvalues are negative,  and the DFE is stable if $\G_{0u}>1$ and $\R_0<1$. 
\end{proof}

\subsection{Complete-infection Equilibrium}
At the CIE, \cref{eq:Jac} becomes
\begin{equation}
J_{CIE}=\left(
\begin{array}{c;{2pt/2pt}c}
A_{CIE} & 0\\ \hdashline[2pt/2pt]
C_{CIE} & D_{CIE}
\end{array}
\right),
\label{eq:Jac_CIE}
\end{equation}
where 
\begin{equation*}
A_{CIE} = \left(
\begin{array}{cccc}
-(\mu_a+\psi) & 0 & \frac{\phi_u}{\G_{0w}} & 0\\
b_f\psi & -\sigma-\mu_{fu} & 0 & 0\\
0 & 0 & -\mu_{fu} & 0\\
b_m\psi & 0 & 0 & -\mu_{mu}
\end{array}
\right),
\end{equation*}
and
\begin{equation*}
D_{CIE} = \left(
\begin{array}{cccc}
-\G_{0w}(\mu_a+\psi) & 0 & \frac{\phi_w}{\G_{0w}} & 0\\
b_f\psi & -\sigma-\mu_{fw} & 0 & 0\\
0 & \sigma & -\mu_{fw} & 0\\
b_m\psi & 0 & 0 & -\mu_{mw}
\end{array}
\right).
\end{equation*}
Because $J_{CIE}$ is a lower triangular block matrix, the eigenvalues of matrix $J_{CIE}$ are the collection of those for matrix $A_{CIE}$ and $D_{CIE}$.

\begin{theorem}[Stability of Complete-infection Equilibrium]\label{thm:th_CIE}
The complete-infection equilibrium $EE^c=(0,A_w^c,0,F_w^c,0,F_{pw}^c,0,M_w^c)$ and \cref{eq:CIE} of the system \cref{eq:ODEa,eq:ODEb,eq:ODEc,eq:ODEd,eq:ODEe,eq:ODEf,eq:ODEg,eq:ODEh} is LAS if $\G_{0w}>1$.
\end{theorem}
The proof, presented in \cref{sec:A2}, is similar to the proof of theorem \cref{thm:th_DFE}.

\subsection{Stability of the Endemic Equilibrium} 
At the EE, \cref{eq:Jac} becomes
\begin{equation*}
J_{EE}=\left(
\begin{array}{c;{2pt/2pt}c}
A_{EE} & B_{EE}\\ \hdashline[2pt/2pt]
C_{EE} & D_{EE}
\end{array}
\right).
\end{equation*}
Unlike the previous two cases, where we have nice upper/lower diagonal block-matrix, in this case, we have a full $8\times 8$ matrix, and the theoretical analysis of the eigenvalues of this matrix is beyond the ability of the authors. However, we are able to numerically verify the following conclusion:

\begin{theorem}[Stability of Endemic Equilibrium]
When the maternal transmission is perfect, $v_w=1$, the endemic equilibrium $EE^*$ (for $\R_0<1$ and $\G_{0w}>1$) is an unstable equilibrium of the system \cref{eq:ODEa,eq:ODEb,eq:ODEc,eq:ODEd,eq:ODEe,eq:ODEf,eq:ODEg,eq:ODEh}. When the maternal transmission is imperfect $v_w<1$, the endemic equilibrium $EE^+$ (for $\R_0>4 v_u v_w$) is a LAS equilibrium, and $EE^-$ (for $4 v_u v_w<\R_0<1$) is an unstable equilibrium.
\end{theorem}

The stability of the three EE is summarized in \cref{tab:table_s}.
\begin{table}[thb]
\centering
\caption{Existence and stability of equilibrium points for \textit{Wolbachia} model \cref{eq:ODEa,eq:ODEb,eq:ODEc,eq:ODEd,eq:ODEe,eq:ODEf,eq:ODEg,eq:ODEh} for both perfect and imperfect maternal transmissions.\cref{eq:EEstar}\cref{eq:EEplus}\cref{eq:EEminus}
\label{tab:table_s}}
\begin{tabular}{m{0.26\textwidth}m{0.12\textwidth}m{0.12\textwidth}m{0.35\textwidth}}
\toprule
 & DFE  & CIE & EE\\ 
\midrule
Perfect maternal \newline transmission\newline ($ v_w=1$) & $\R_0<1$  \newline  $\G_{0u}>1$  \newline  (LAS)  & $\G_{0w}>1$ \newline  (LAS) & $\R_0<1$ and $\G_{0w}>1$ \newline $\circ$ $r_{wu}=r_{wu}^*=(1-\R_0)/\R_0$ ;\newline $EE^*$ (unstable) \\ 
\midrule
Imperfect maternal \newline transmission\newline ($ v_w<1$) & $\R_0<1$ \newline $\G_{0u}>1$ \newline (LAS) & \emph{N/A} & 
\underline{$4 v_u v_w<\R_0<1$}\newline $\circ$ $r_{wu}=r_{wu}^+$ ; $EE^+$ (LAS) \newline $\circ$ $r_{wu}=r_{wu}^-$ ; $EE^-$ (unstable) \newline \underline{$\R_0>1$}\newline $\circ$ $r_{wu}=r_{wu}^+$; $EE^+$ (LAS) \\
\bottomrule
\end{tabular}
\end{table}

\section{Bifurcation Analysis}\label{sec:sec_5}
\begin{figure}[tb]
\centering
\includegraphics[width =0.6\textwidth]{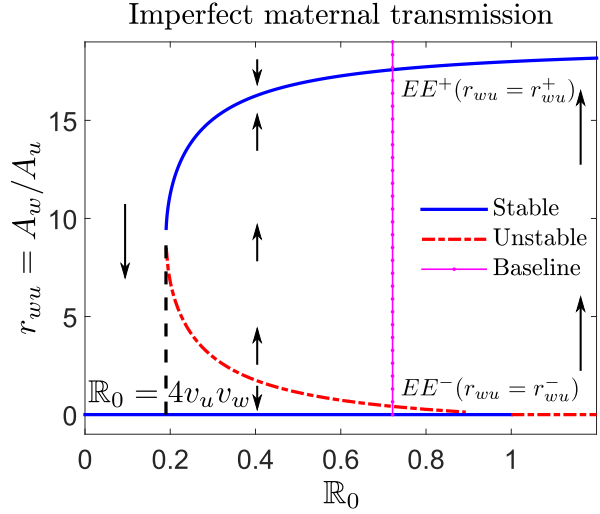}\label{fig:bifur}
\caption{Backward bifurcation diagram for imperfect maternal transmission ($v_w=0.95$). Both the stable equilibrium points DFE (horizontal line for $0<\R_0<1$) and stable branch of the EE (upper branch of the ``fork'') are represented by solid blue curves. The unstable branch of the EE (lower branch of the ``fork'') is represented by the dashed red curve, which is the threshold condition for having stable \textit{Wolbachia} endemic. Two branches meet at $\R_0=4 v_u v_w (=0.19)$. The baseline case ($\R_0=0.721$) is marked by the vertical magenta line with dots. The arrows indicate the direction of the phase flows.}
\end{figure}

\textit{Wolbachia} infection is not naturally found for \textit{Ae. aegypti} mosquitoes, suggesting that $\R_0<1$ for wild mosquitoes. We need to introduce infected mosquitoes into the environment for the system to surpass the threshold condition. In the case of imperfect maternal transmission, this threshold condition is described by the backward bifurcation diagram in \cref{fig:bifur}. The y-axis of the diagram is the ratio $r_{wu}$ introduced in \cref{sec:sec_34}. The DFE is marked by a horizontal solid line, where $r_{wu}=0$ for $0<\R_0<1$. The baseline case (in \cref{tab:parameter}) is highlighted by a vertical magenta line with dots.

When $0<\R_0<4 v_u v_w$ is small, DFE is the only steady state, and it is globally stable. When $\R_0>1$, the only stable steady state is the upper branch of EE. At $\R_0=4 v_u v_w$, there appear three equilibrium states, and in the interval $4 v_u v_w<\R_0<1$ the DFE is stable, the middle $EE^-$ is unstable, and the upper $EE^+$ state is stable. The lower branch state $EE^-$ is the threshold condition for having endemic \textit{Wolbachia}: Below the threshold state $EE^-$, the \textit{Wolbachia}-infected mosquitoes that have been introduced to the environment are wiped out by the wild population, and the system goes back to DFE; Above the threshold $EE^-$, the \textit{Wolbachia}-infected mosquitoes are able to gradually invade into the wild environment, and at some point, the endemic equilibrium $EE^+$ is achieved, where both infected and uninfected mosquitoes are coexisted in the environment. 

In \cref{fig:bifur2} we define the vertical axis by the percentage of infected females, including both the infected nonpregnant females $F_w$ and infected pregnant females $F_{pw}$.  We have found this representation provides a more intuitive understanding of the bifurcation process since as the maternal vertical transmission rate decreases, the threshold condition (unstable EE) increases and the prevalence of infection (percentage of infected females) decreases. It is possible to establish a stable endemic state over a wide range of $\R_0$ values as long as a significant fraction of the mosquito population is infected and the maternal vertical transmission rate is high (e.g. $ v_w  > 0.99$).

\begin{figure}[t]
\centering
\includegraphics[height =0.35\textwidth]{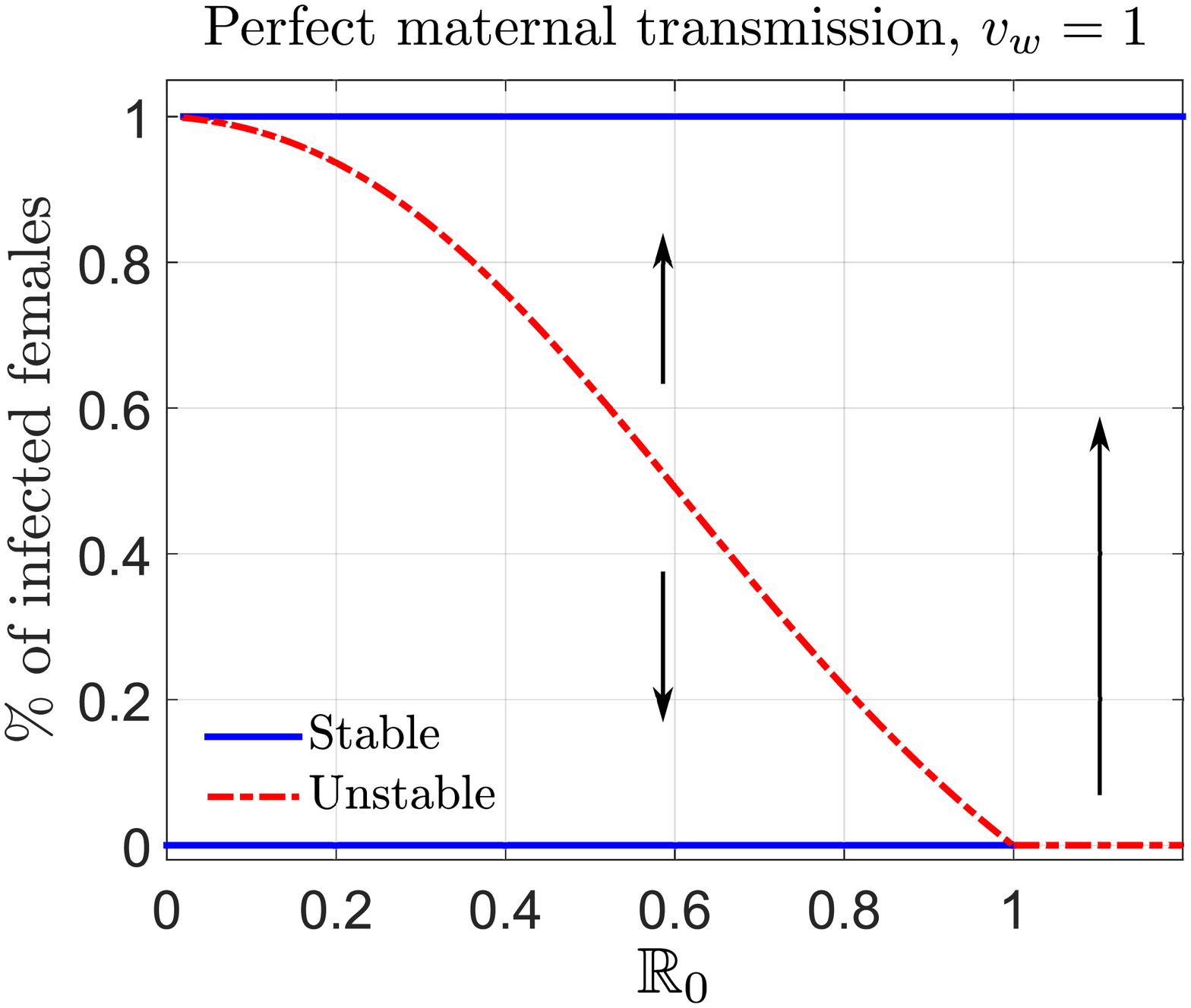}
\includegraphics[height =0.35\textwidth]{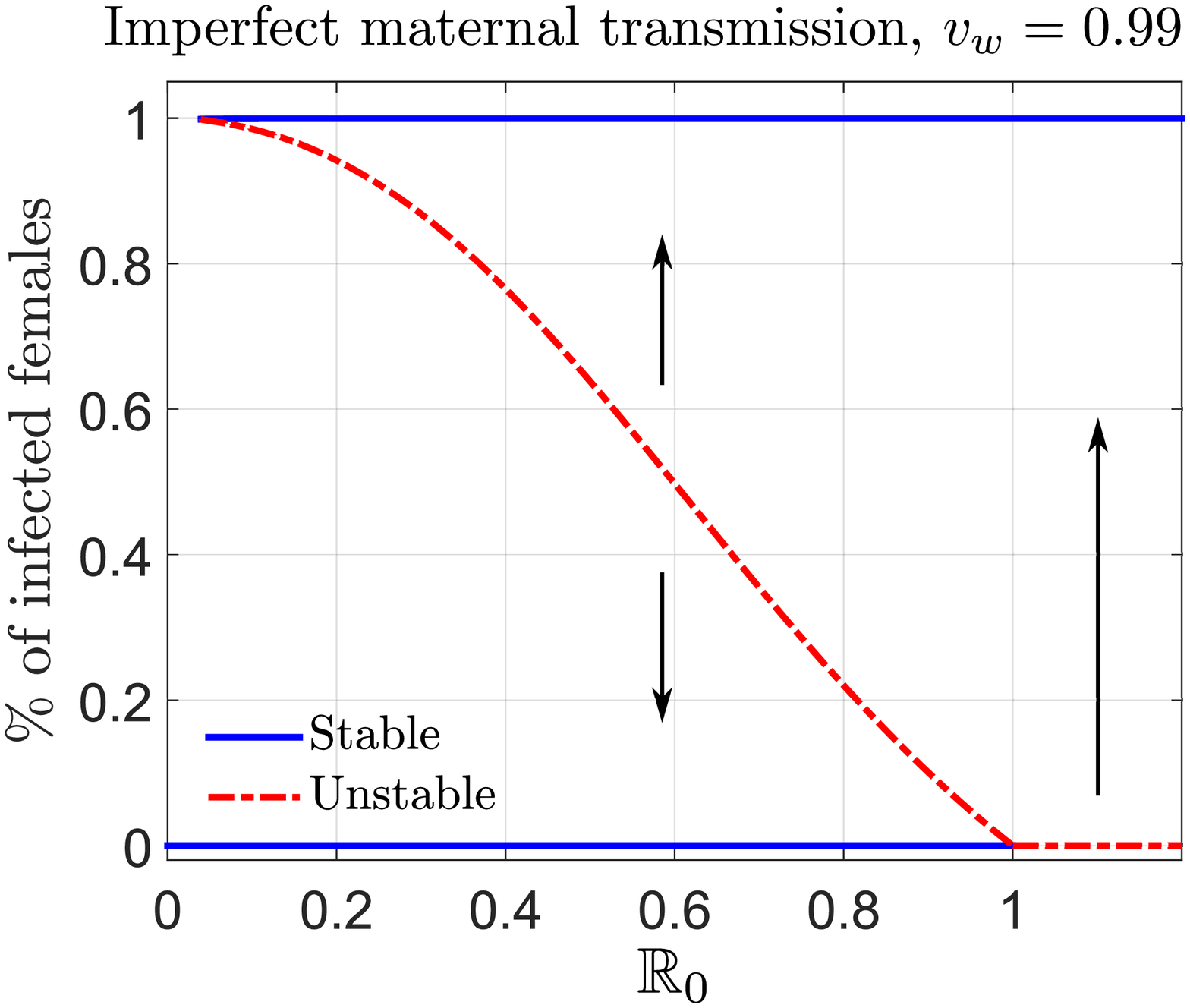}\\
\includegraphics[height =0.35\textwidth]{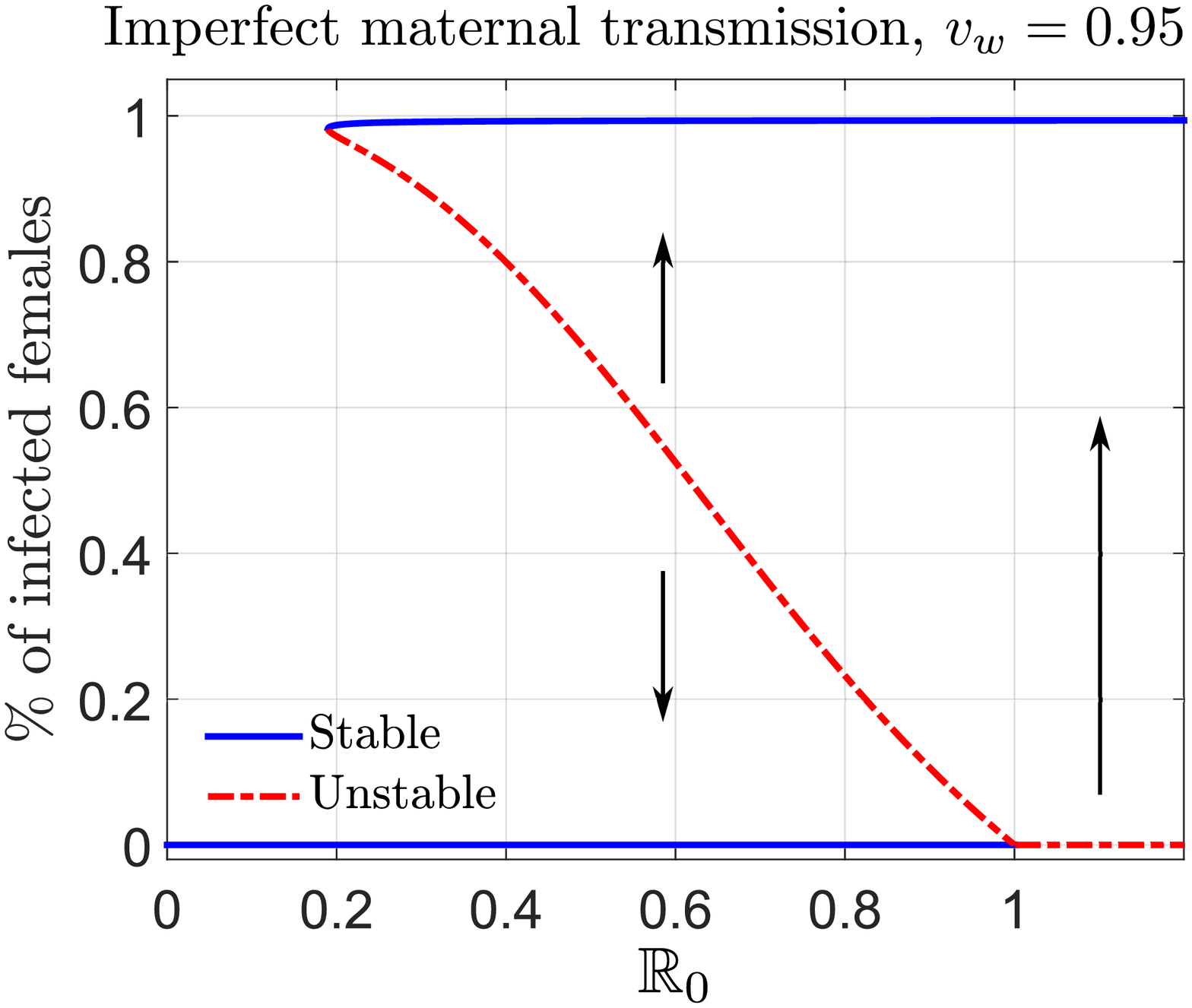}
\includegraphics[height =0.35\textwidth]{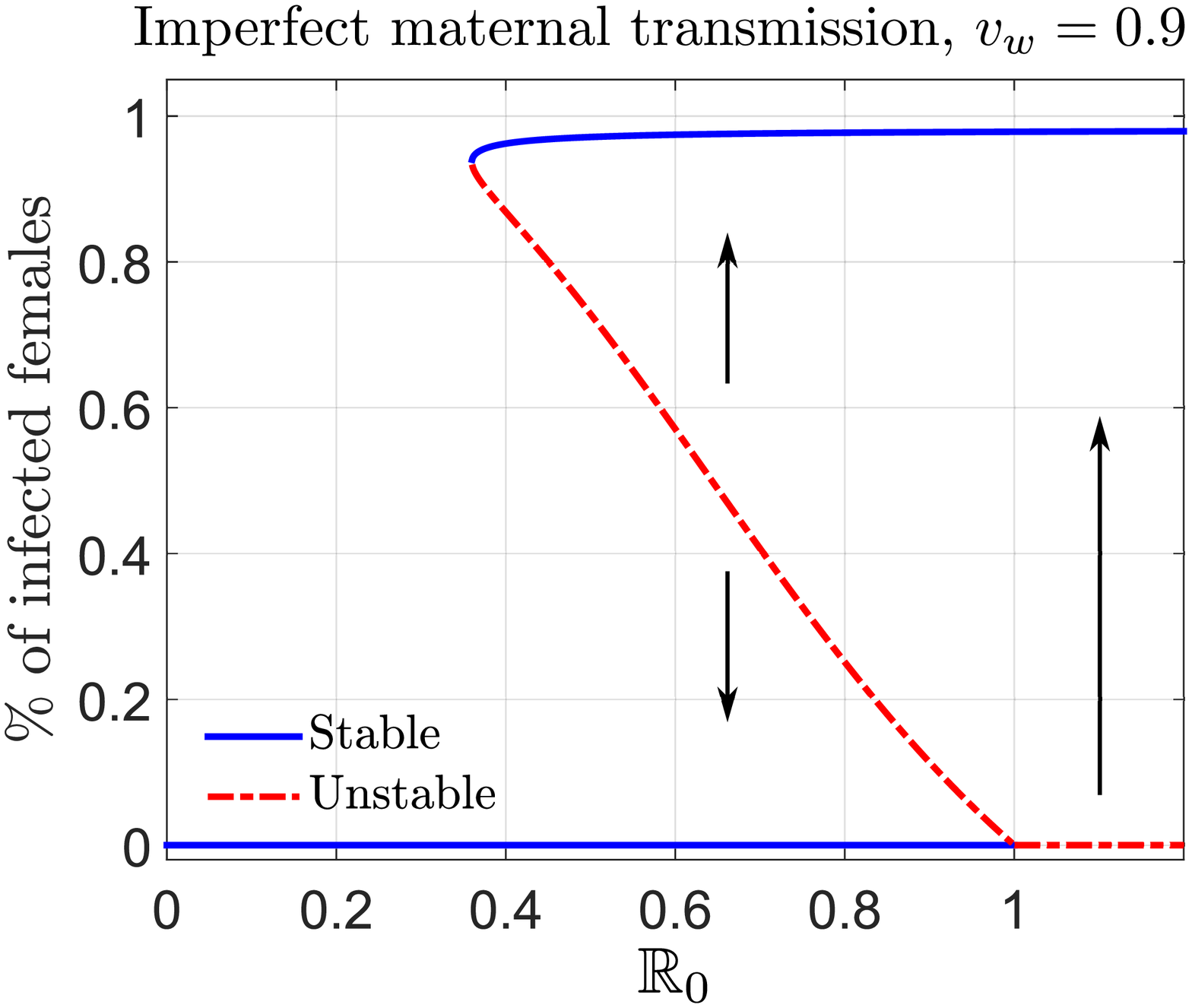}

\caption{The threshold conditions for having a stable endemic state are characterized by these bifurcation diagrams for different maternal transmission rates.  In our simulations, the threshold conditions are determined for the baseline case at $\R_0=0.721$. The red dashed curves represent the unstable EE, and the solid blue curves correspond to the stable EE.  As the maternal transmission rate decreases, the threshold condition  increases and the prevalence of infection decreases. When maternal transmission rate is high, it is possible to establish a stable endemic state for a wide range of $\R_0$ values. However, when the transmission rate is low (e.g. $v_w=0.9$), due to the global stability of the DFE, a stable endemic state is unattainable for small $\R_0$ values.\label{fig:bifur2}}
\end{figure}

\section{Sensitivity Analysis}\label{sec:sec_6}
The baseline values in \cref{tab:parameter} represent our best-guess estimates of the model parameters. It is difficult to obtain good estimates of the key fitness parameters \cite{mcgraw2013beyond}, and we investigate the model dynamics over a wide range of feasible parameters to help better understand the model response under different  assumptions. Also, the scalar model parameters are approximations of the mean of an underlying distribution.  For example, the fitness of mosquitoes (lifespan or egg laying rate) are not the same for every mosquito.  We quantify the significance of these parameters in the model predictions using local and extended sensitivity to measure the relative change in the output quantities of interests (QOIs) with respect to the input parameters of interests (POIs). 

Following the framework in \cite{chitnis2008determining}, we define the normalized relative sensitivity index of a QOI, $q(p)$, with respect to the POI, $p$, as
\begin{equation}\label{eq:SI_def}
\mathcal{S}_p^q:=\frac{p}{q}\times\frac{\partial q}{\partial p} 
\end{equation}
over the plausible range of parameter $p$. The relative sensitivity index $\mathcal{S}_p^q$ measures the percentage change in the QOI given the percentage change in an input POI, that is, if parameter $p$ changes by $x\%$, then quantity $q$ changes by $\mathcal{S}_p^q\times x \%$. The sign of $\mathcal{S}_p^q$ determines if the response is increasing or decreasing.  When evaluated at the baseline parameter values,  $p=\hat{p}$  and $\hat q = q(\hat{p})$, then  
\begin{equation*}
\mathcal{S}_{\hat p}^q:=\mathcal{S}_p^q\,\Big|_{p=\hat{p}}=\frac{\hat{p}}{\hat q}\times\frac{\partial q}{\partial p}\,\Big|_{p=\hat{p}}
\end{equation*}
is called the local relative sensitivity index of $q$ at $\hat{p}$.

The fitness cost (on lifespan and egg laying rate) and maternal transmission rate are two key factors to the potential success of \textit{Wolbachia} infection being established in a wild mosquito population \cite{walker2011wmel}. We consider POIs that  measure the loss of fitness caused by \textit{Wolbachia} infection and define 
\begin{itemize}
\item[-] fitness cost on lifespan $p_{\mu_f}:=\left(\mu_{fu}^{-1}-\mu_{fw}^{-1}\right)\Big/\left(\mu_{fu}^{-1}\right)$, the fractional reduction in female's lifespan caused by \textit{Wolbachia}  infection,
\item[-] fitness cost on egg laying rates $p_{\phi}:=\left(\phi_u-\phi_w\right)\big/\left(\phi_u\right)$, the fractional reduction in the egg laying rate caused by \textit{Wolbachia} infection, and
\item[-] $p_{v_w}:=v_w$, the maternal transmission rate of \textit{Wolbachia}-infected females to offspring. 
\end{itemize} 
Meanwhile, we choose the QOIs that capture important aspects of the epidemic: 
\begin{itemize}
\item[-] $q_{th}$, the threshold condition reflected in the percentage of infected females needed for having stable endemic \textit{Wolbachia},
\item[-] $q_{T_{90}}$, time to $90\%$ infection in the female population, which measures the speed of infection spreading,
\item[-] $q_{ep}$, the endemic prevalence reflected in the percentage of infected females at endemic steady-state.
\end{itemize}
The results of the pairwise local sensitivity indices for each QOI against POI are listed in \cref{tab:SA_indices}, and the corresponding extended sensitivity analysis plots are presented in \cref{fig:ex1_SA}. In \cref{fig:fig_ex1b}, we have assigned 70\% infection as initial condition for the simulations.

\begin{table}[tbhp]
\caption{Local relative sensitivity indices $\mathcal{S}_{\hat p}^q$: a higher fitness cost or lower maternal transmission rate makes it the less efficient to establish \textit{Wolbachia} infection, which is reflected through higher threshold, longer spreading process and smaller final infection prevalence. The maternal transmission rate has the largest impact on all three quantities: if the maternal transmission increases by 1\%, the threshold condition decreases by 4.36\%.\label{tab:SA_indices}} 
\centering
\begin{tabular}{lccc}
\toprule
 & $\hat{p}_{\mu_f}$ & $\hat{p}_{\phi}$ & $\hat{p}_{v_w}$\\
\midrule
threshold condition ($q_{th})$ & 0.342 & 0.662 & $-4.36$ \\
time to $90\%$ infection ($q_{T_{90}}$) & 0.146 & 0.210 & $-5.51$ \\
endemic prevalence ($q_{ep}$) & $-7.91\times10^{-4}$ & $-1.53\times 10^{-4}$ & 0.218\\
\bottomrule
\end{tabular}
\end{table}

\begin{figure}[tbhp]
\centering
\subfloat[$\mathcal{S}_{\hat{p}_{\mu f}}^{q_{th}}=0.342$, $\mathcal{S}_{\hat{p}_{\phi}}^{q_{th}}=0.662$ and $\mathcal{S}_{\hat{p}_{v_w}}^{q_{th}}=-4.36$. At baseline case, approximately 30\% infection is needed for having a stable endemic state, and there are linear trends between threshold condition and fitness cost/maternal transmission rate.]{\label{fig:fig_ex1a}\includegraphics[width =0.47\textwidth]{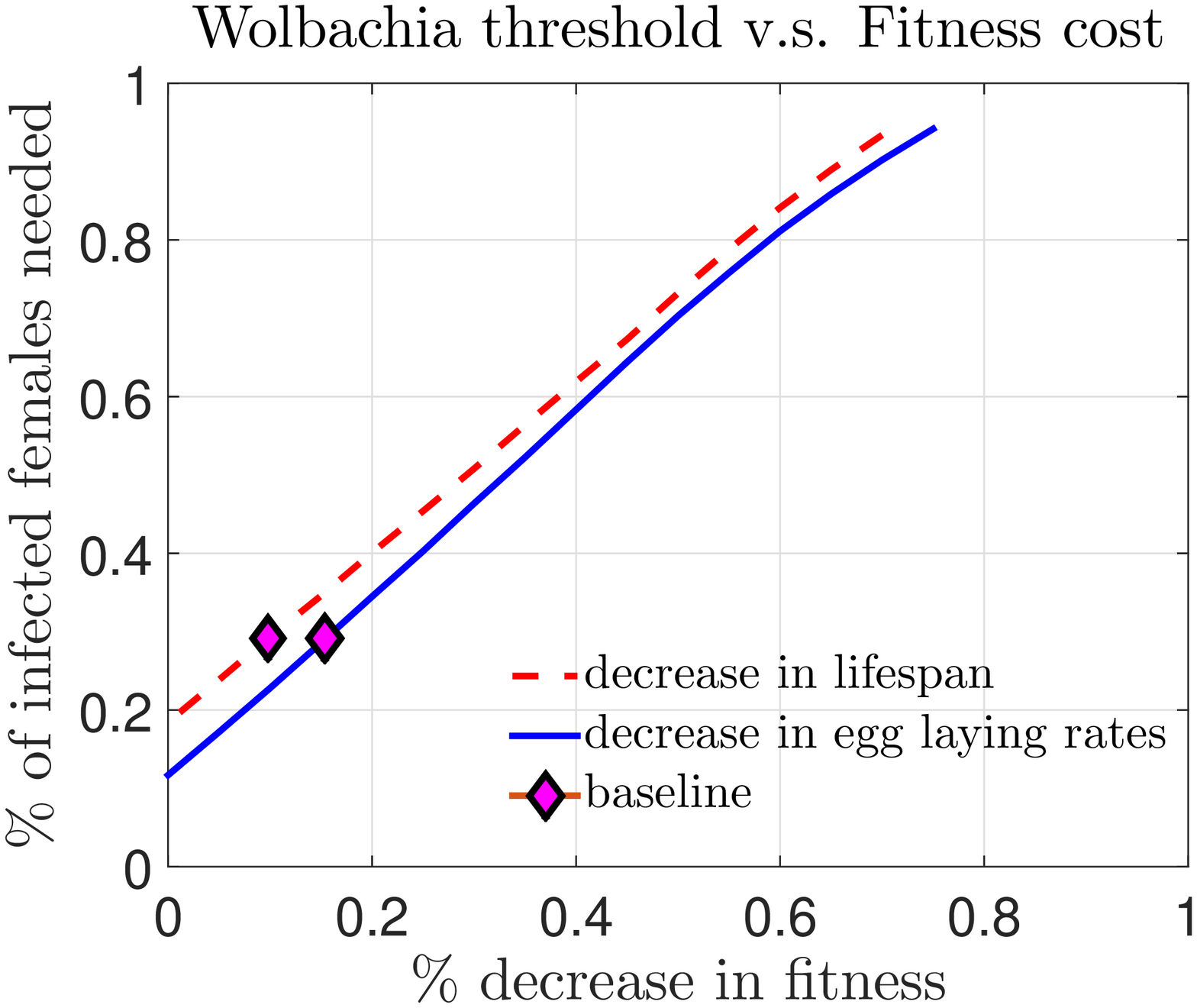}\hspace{0.03\textwidth}
\includegraphics[width =0.47\textwidth]{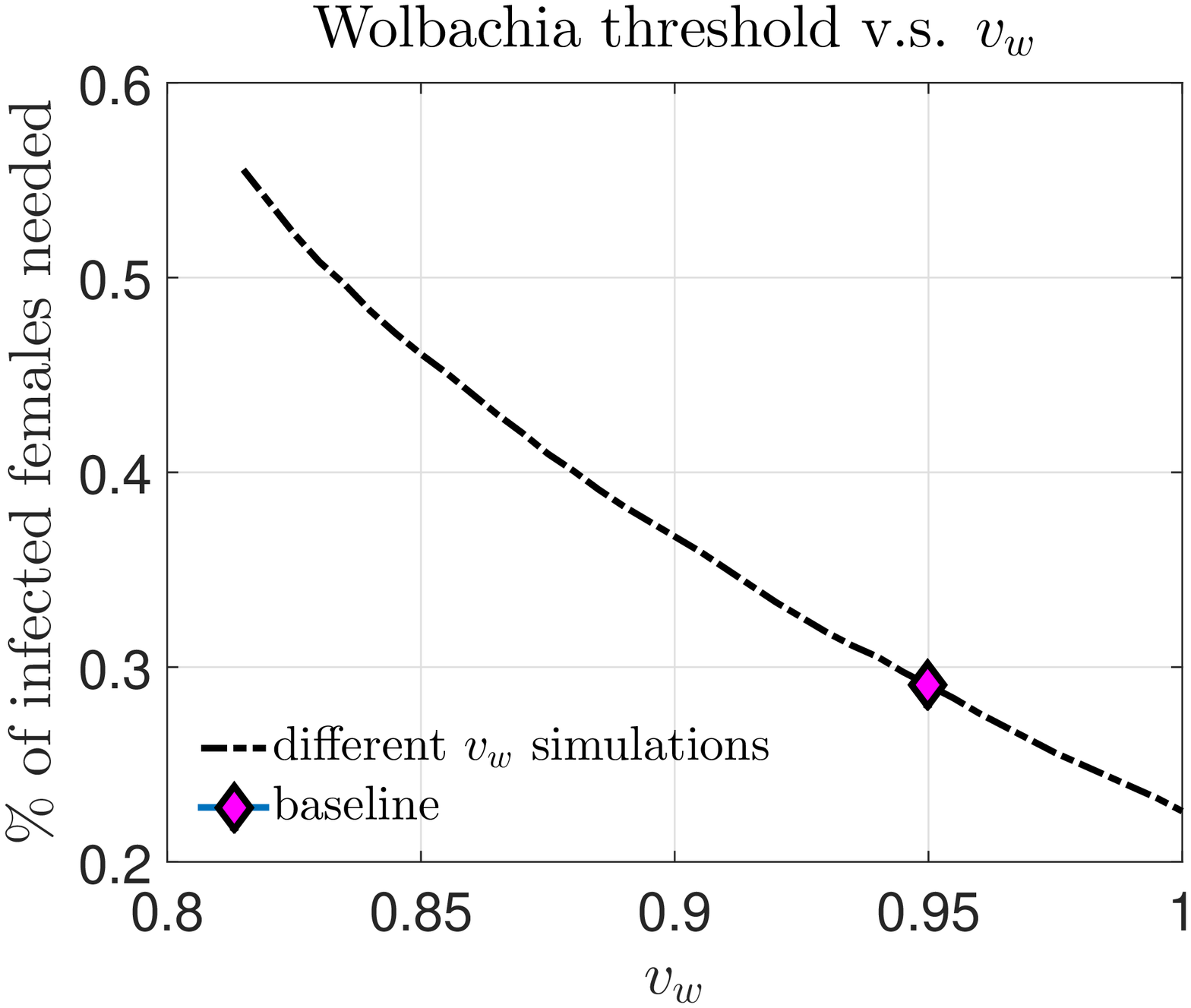}}\newline
\subfloat[$\mathcal{S}_{\hat{p}_{\mu f}}^{q_{T_{90}}}=0.146$, $\mathcal{S}_{\hat{p}_{\phi}}^{q_{T_{90}}}=0.210$ and $\mathcal{S}_{\hat{p}_{v_w}}^{q_{T_{90}}}=-5.51$. At baseline case, the fitness cost is about 13\%, and it would take about 50 days to achieve a stable endemic state. When this cost is increased to 45\%, this time increases to almost half a year. Also, when maternal transmission rate is decreased below 0.83, it is not practical to establish a stable endemic state.]{\label{fig:fig_ex1b}
\includegraphics[width =0.47\textwidth]{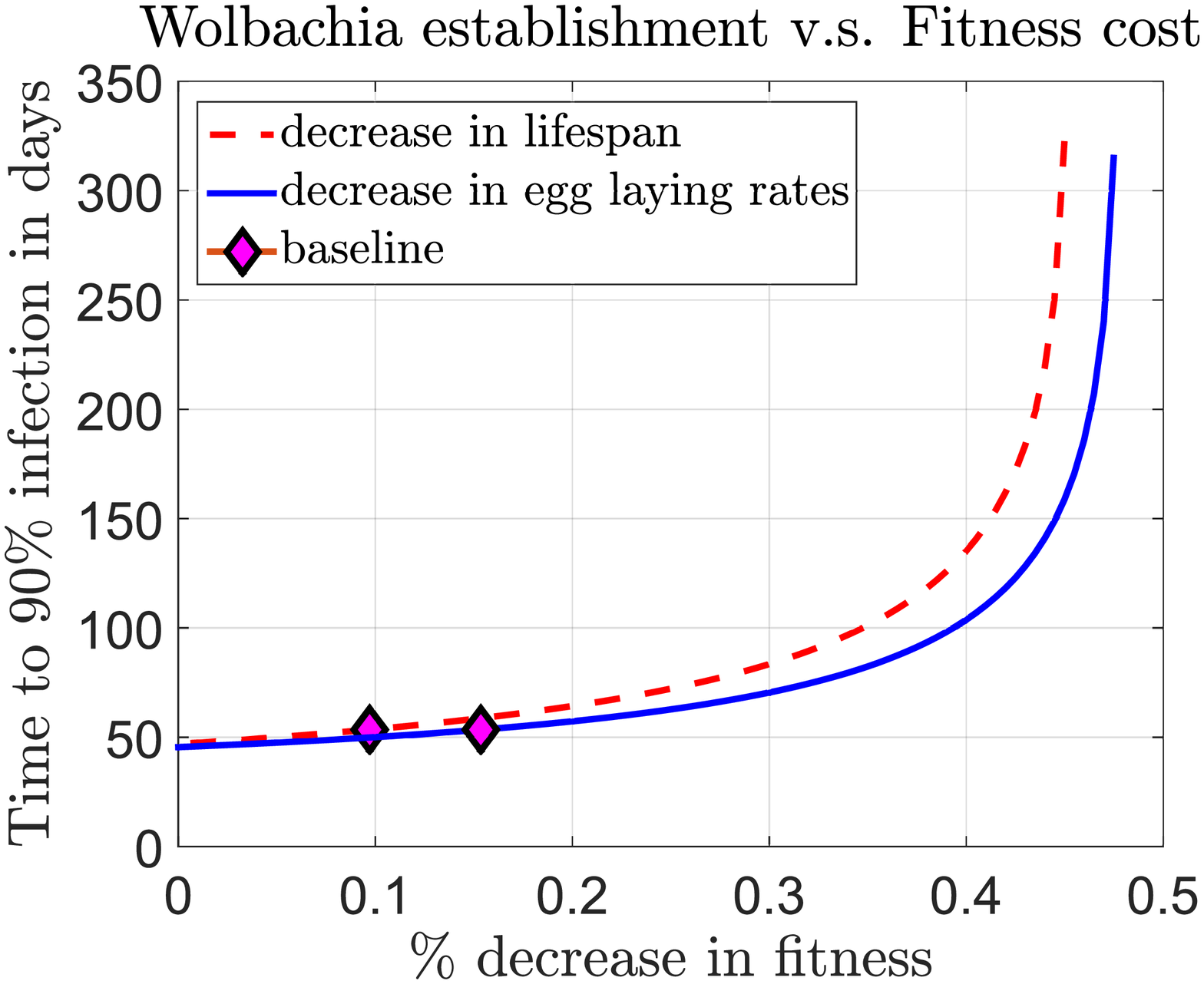}\hspace{0.03\textwidth}
\includegraphics[width =0.47\textwidth]{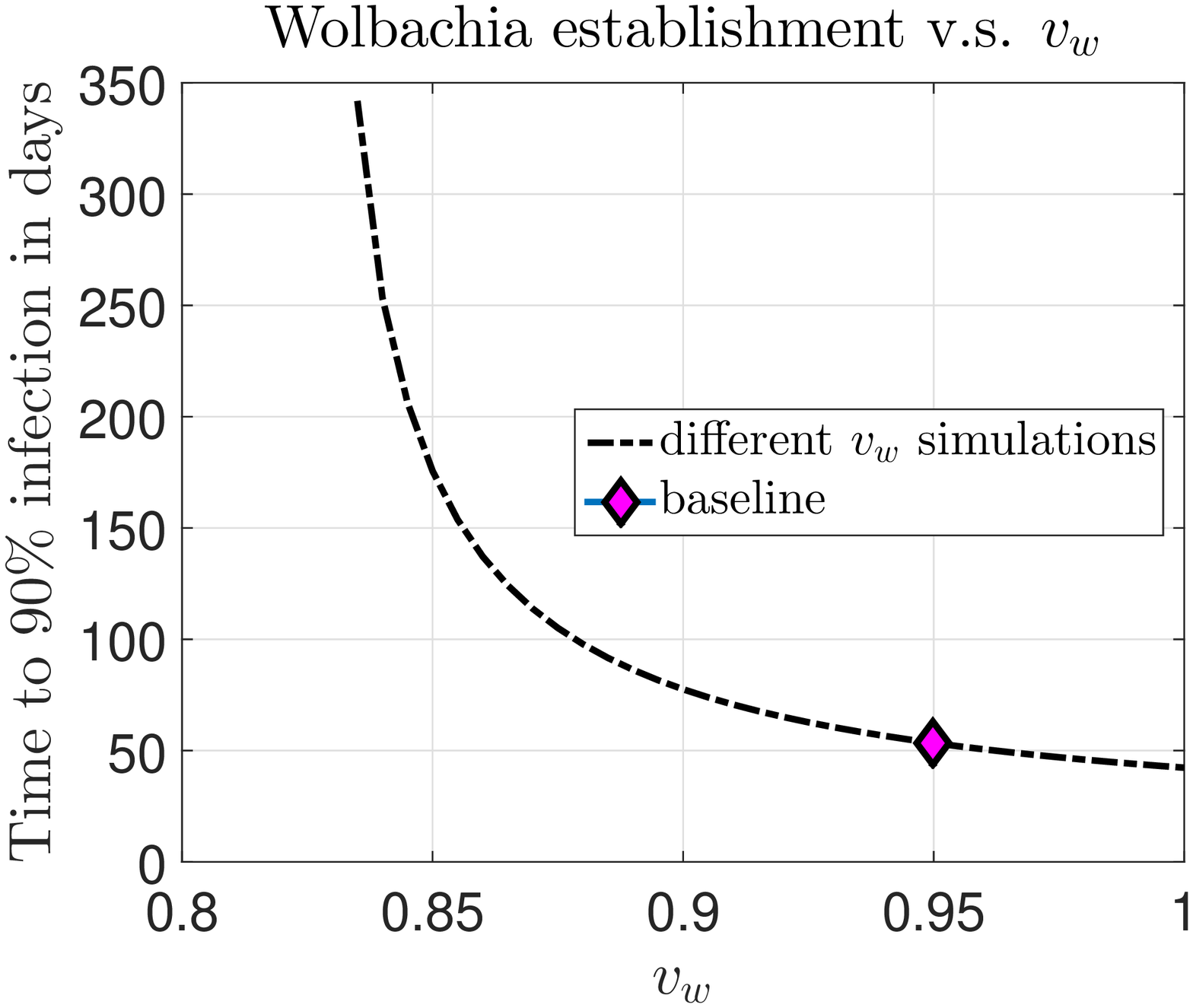}}\newline
\subfloat[$\mathcal{S}_{\hat{p}_{\mu f}}^{q_{ep}}=-7.91\times10^{-4}$, $\mathcal{S}_{\hat{p}_{\phi}}^{q_{ep}}=-1.53\times 10^{-4}$ and $\mathcal{S}_{\hat{p}_{v_w}}^{q_{ep}}=0.218$. Near the baseline case, the prevalence mainly depends on the maternal transmission rate and is not sensitive to either lifespan or egg laying rates. As the fitness cost increases, the lifespan becomes more important than egg laying rate in determining the prevalence.]{\label{fig:fig_ex1c}
\includegraphics[width =0.47\textwidth]{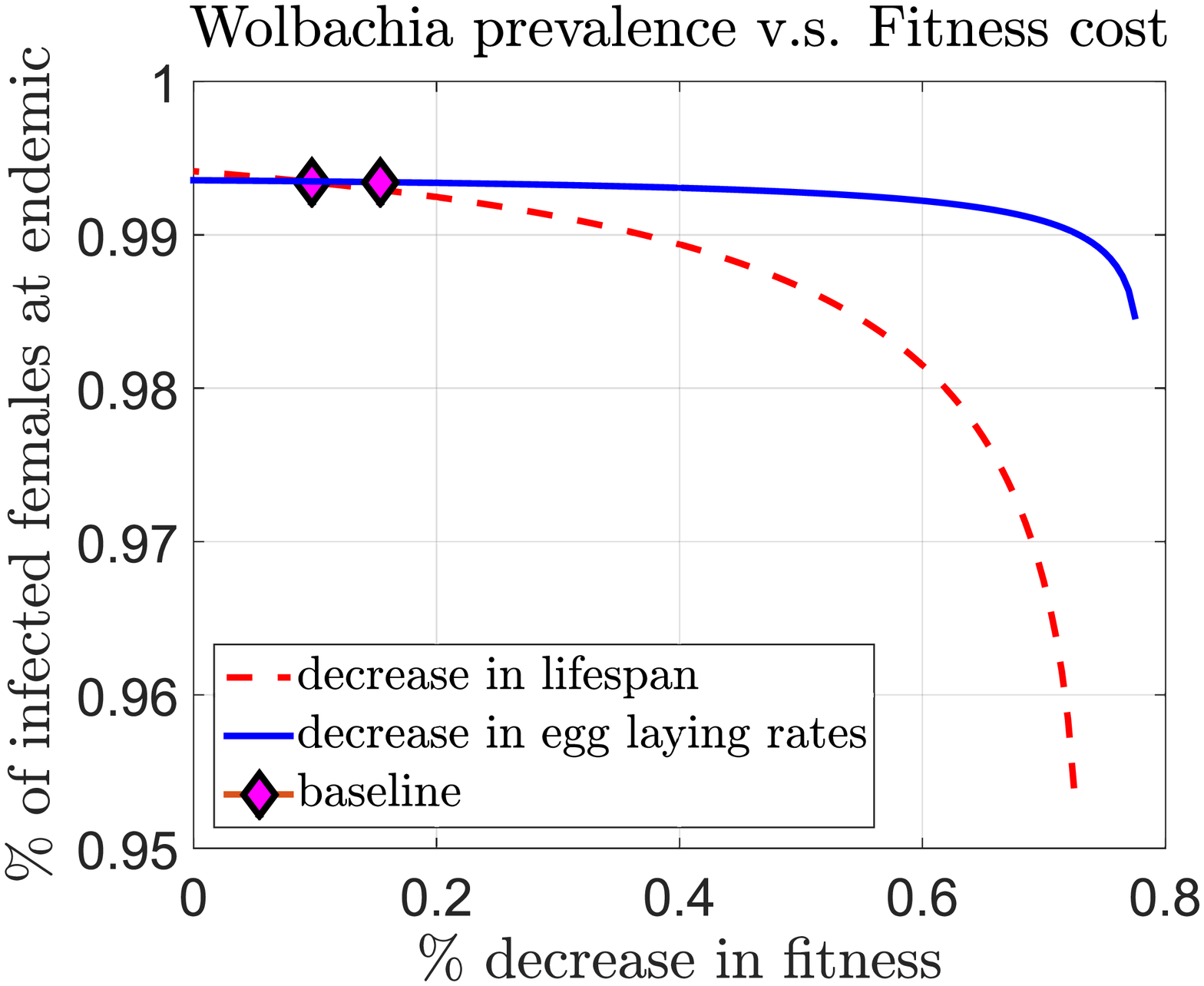}\hspace{0.03\textwidth}
\includegraphics[width =0.47\textwidth]{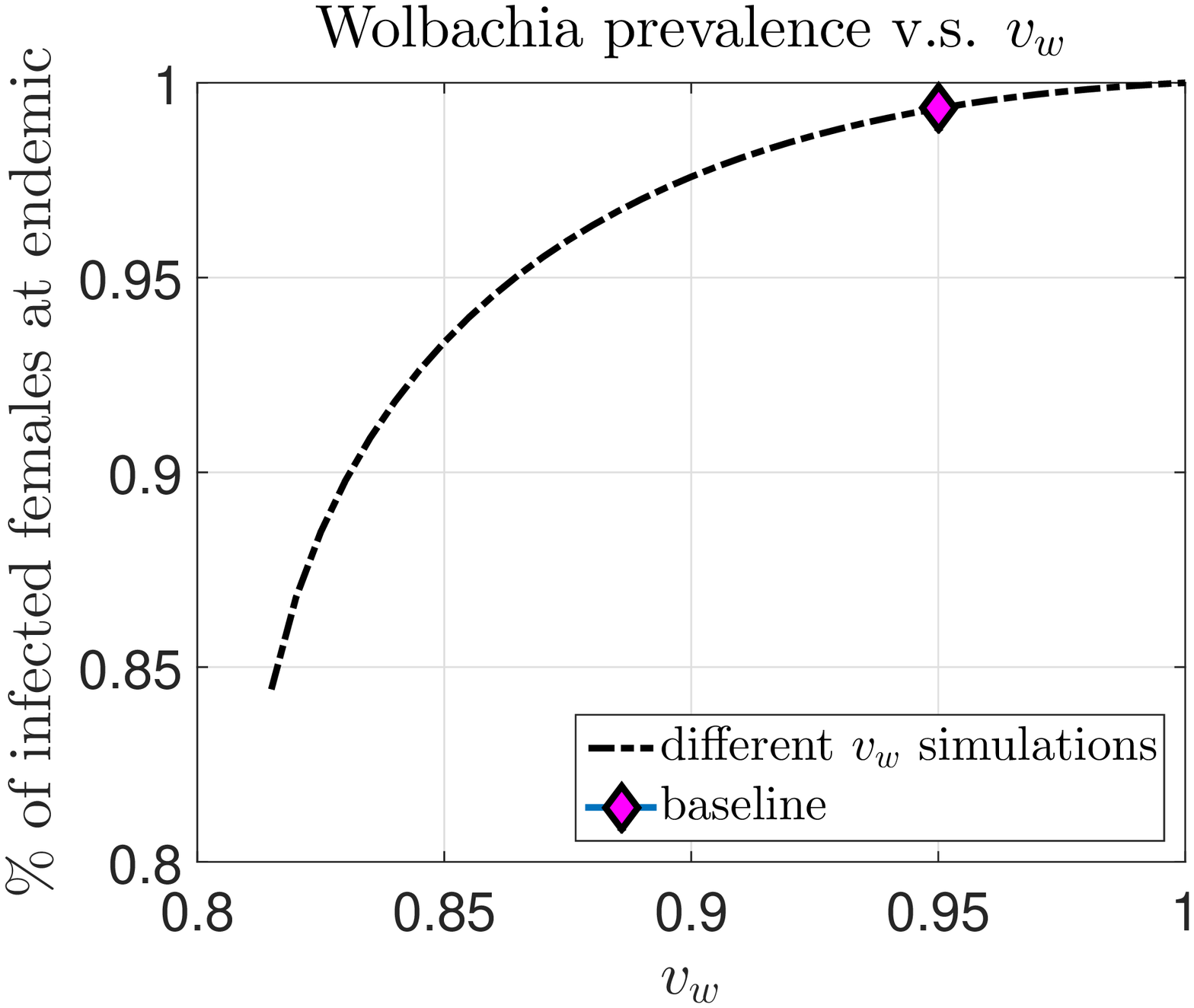}}

\caption{Sensitivity analysis: study how the fitness cost (reduction in lifespan and egg laying rate) and maternal transmission rate impact threshold condition, speed of establishing a stable endemic state and prevalence at endemic state, respectively. The extended sensitivity analysis curves show the changes in QOIs over the full range of the POIs. The diamonds indicate the baseline case used in the simulations.}\label{fig:ex1_SA}
\end{figure}

The sensitivity indices in \cref{tab:SA_indices} confirm that a higher fitness cost or lower maternal transmission rate makes it the less efficient to establish \textit{Wolbachia} infection, which is reflected through higher threshold, longer spreading process and smaller final infection prevalence. The maternal transmission rate has the largest impact on all three QOIs. 

In the extended sensitivity anlysis, we vary one POI at a time over the full parameter range and fix other parameters in the model. The baseline values (listed in \cref{tab:parameter}) gives about 13\% fitness cost for wMel strain \textit{Wolbachia} infection. \Cref{fig:fig_ex1a} shows that approximately 30\% of the female population should be infected to have a stable \textit{Wolbachia} infection, and \cref{fig:fig_ex1b} predicts that it takes about 50 days for the infection to grow from 70\% to 90\%. 

The fitness cost of  wMelPop \cite{mcmeniman2009stable} and wMelPop-CLA \cite{walker2011wmel} can be as large as 45\%. Our model predicts that for these strains, the threshold infection percentage is as high as 65\% and it takes about half a year to achieve 90\% infection in the population.  Our model does not include many real-world effects, such as the diffusion of mosquitoes in and out of the control areas, and the high fitness cost of these strains indicates that it would require multiple recurring releases of \textit{Wolbachia}-infected mosquitoes to establish a self-sustaining infected population. 

Near the baseline case, the prevalence of \textit{Wolbachia} infection at endemic steady-state is insensitive to the fitness cost and is mainly decided by the maternal transmission rate $v_w$ (\cref{fig:fig_ex1c}). As the fitness cost increases, the lifespan becomes more important than egg laying rate in determining the prevalence.

\section{Comparing Mitigation Strategies}\label{sec:sec_7}
Because the threshold condition is characterized by a minimal fraction of mosquitoes that are infected, the number of infected mosquitoes that must be released to exceed the condition can be reduced by first reducing the population of uninfected mosquitoes.   We consider using pre-release mitigation strategies based on \emph{residual spraying} to reduce both the adult and aquatic stage uninfected mosquitoes, \emph{larval control} to reduce the number of uninfected eggs, larvae, and pupae before the release, \emph{sticky gravid traps/ovitrap} to reduce the number of uninfected pregnant female mosquitoes, and \emph{acoustic attraction} to reduce the number of uninfected male mosquitoes (\cref{tab:IMM}).

\begin{table}[tbhp]
\centering
\caption{Pre-release mitigation approaches and the corresponding effectivenesses for mosquito control. \label{tab:IMM}}
\begin{tabular}{lll}
\toprule
Approach  & Target & Effectiveness \\
\midrule
Residual spraying & Adults \& larvae  & Adults $90\%$, larvae $40\%$ \cite{nguyen2009evaluation,ritchie2002dengue} \\
Larval control & Aquatic-stage & Eggs 50\%, larvae 50\% \cite{promsiri2006evaluations} \\
Sticky gravid traps/ovitrap  & Pregnant females & Pregnant females 75\% \cite{barrera2014use, ritchie2009lethal} \\
Acoustic attraction 	&  Males  & Males 80\% \cite{belton1994attractton} \\
\bottomrule
\end{tabular}
\end{table}

\emph{Indoor/perifocal residual spraying} Since \textit{Ae. aegypti} is an anthropophilic species of mosquito that breeds inside or near houses, indoor residual spraying, which is applied on cryptic resting sites inside premises, and perifocal residual spraying, which is applied to the external building and ornamental plant surfaces, have been used to target harboring or ovipositing adult mosquitoes effectively \cite{nguyen2009evaluation,ritchie2002dengue}. Adulticide can also act as a larvicide when applied to the inner surfaces of receptacles such as vehicle tires. Since larvae occupy about 11\% of the aquatic-stage population at DFE, we take 40\% reduction in larvae as 5\% reduction in overall aquatic-stage population in our numerical simulations.

\emph{Larval control} to reduce the uninfected eggs, larvae, and pupae to increase the carrying capacity for the new infected eggs. Comprehensive larval control that targets water storage using insecticide, biologicals (e.g. larvivorous copepods) and container removal has been successfully employed in field trials for small communities to reduce dengue incidence, and it has been recommended as sustained management of aquatic-stage mosquitoes \cite{achee2015critical}. Similarly, since eggs and larvae take about 89\% and 10\% of the aquatic-stage population, respectively \cite{koiller2014aedes}, we take the effectiveness as 50\% reduction in the overall aquatic-stage population.

\emph{Sticky gravid traps/ovitrap} to reduce the uninfected pregnant females before releasing the infected ones. The use of lethal ovitraps or sticky gravid traps is a customized strategy that attracts and kills female mosquitoes as they lay eggs \cite{barrera2014use}.  

\emph{Acoustic attraction} to reduce the uninfected male population before releasing the infected males. Since male mosquitoes use sound as a guide to seek females for mating, the use of sound generated by audio oscillators or recording of female mosquitoes can be used to selectively attract and kill male mosquitoes. This approach could be a useful methodology to increase the ratio of the released males to the wild males at the release time \cite{belton1994attractton}. 

In practice, some of the these pre-release mitigation methods, such as residual spraying or larvicide, may have sustained low-level efficacy over a long period of time. In our numerical simulations, we have assumed that all the mitigation stops once the release starts, and the released infected mosquitoes will not be affected. We process the pre-release mitigation step as an adjustment to the initial condition of the system.

Our simulations address three integrated mitigation strategies:
\begin{itemize}
\item[Q1:] Is it better to release infected males, nonpregnant females and/or pregnant females?
\item[Q2:] Which pre-release strategies are the most effective? 
\item[Q3:] Is it better to release all the infected mosquitoes at once, or it is better to repetitively release smaller batches?

\end{itemize} 

\subsection{Q1: What is the best mix of infected mosquitoes to release?}

Releasing only infected male mosquitoes, similar to sterile insect releases, can reduce the mosquito populations. Infected males act like adulticide: they sterilize the natural females and reduce the population size. However, this approach requires long-term repetitive releases, hence it is not a self-sustained mitigation strategy.   
To establish an sustained endemic equilibrium, we release both infected males and females to create stable endemic \textit{Wolbachia}. 
We release  $2X=1.8 F_{pu}^0$ mosquitoes, where $F_{pu}^0$ is the number of pregnant females at DFE.  We also assume that we release the same number of males ($X$ males) and females ($X$ females), since there birth rates are almost equal \cite{tun2000effects}.
We compare the approaches:

\emph{Pregnant Female Release (PFR) Approach} releasing infected males, $M_w$, and pregnant females, $F_{pw}$, from the same container.  (When males and females are stored at the same place, nearly all the females become pregnant by the time of the release.)

\emph{Nonpregnant Female Release  (NPFR) Approach} releasing infected males, $M_w$, and nonpregnant females, $F_w$, from different containers.  In this approach, the females are separated from the males shortly after birth. 

\begin{figure}[t]
\centering
\includegraphics[width =0.48\textwidth]{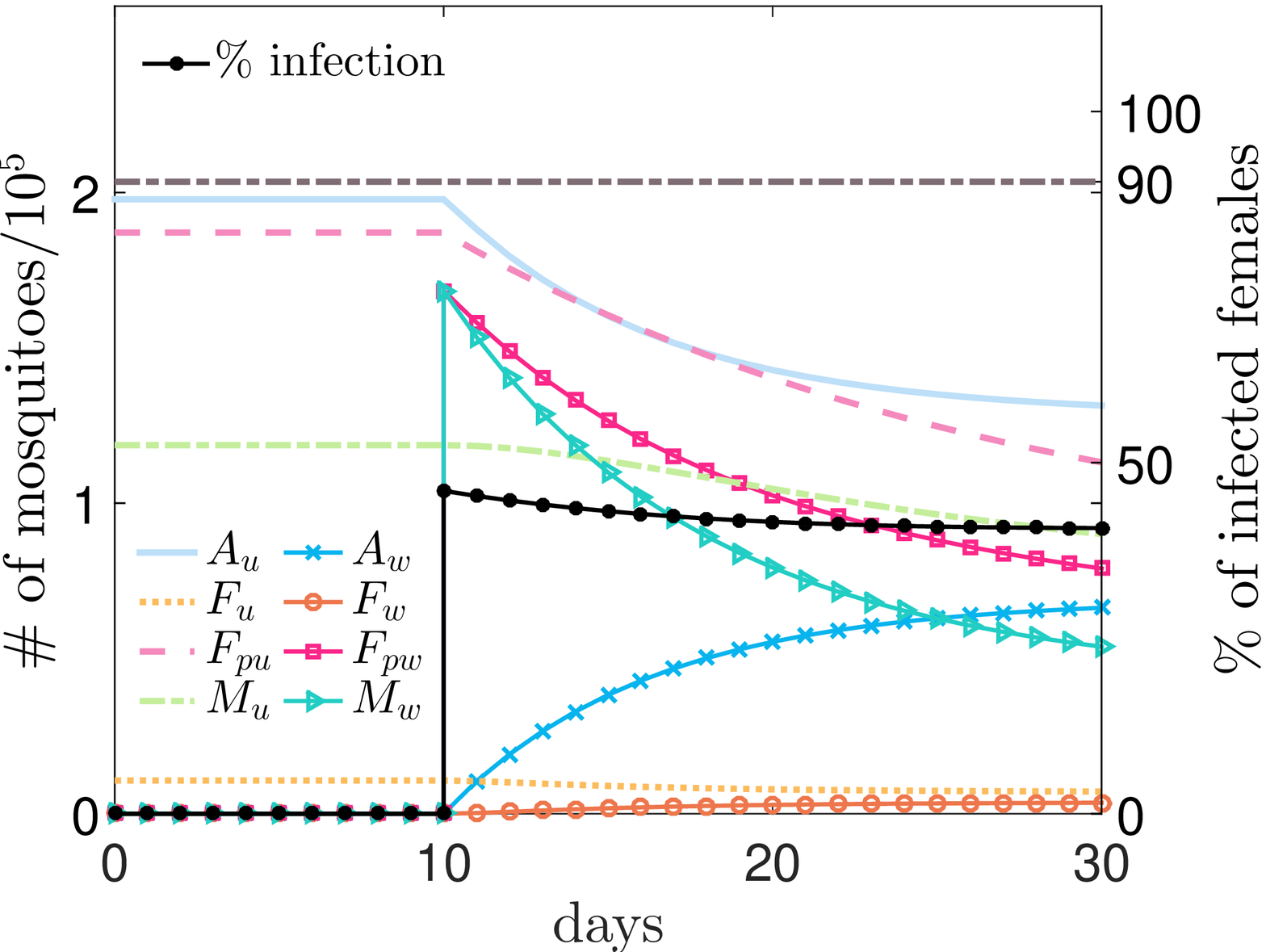}\hspace{0.03\textwidth}
\includegraphics[width =0.48\textwidth]{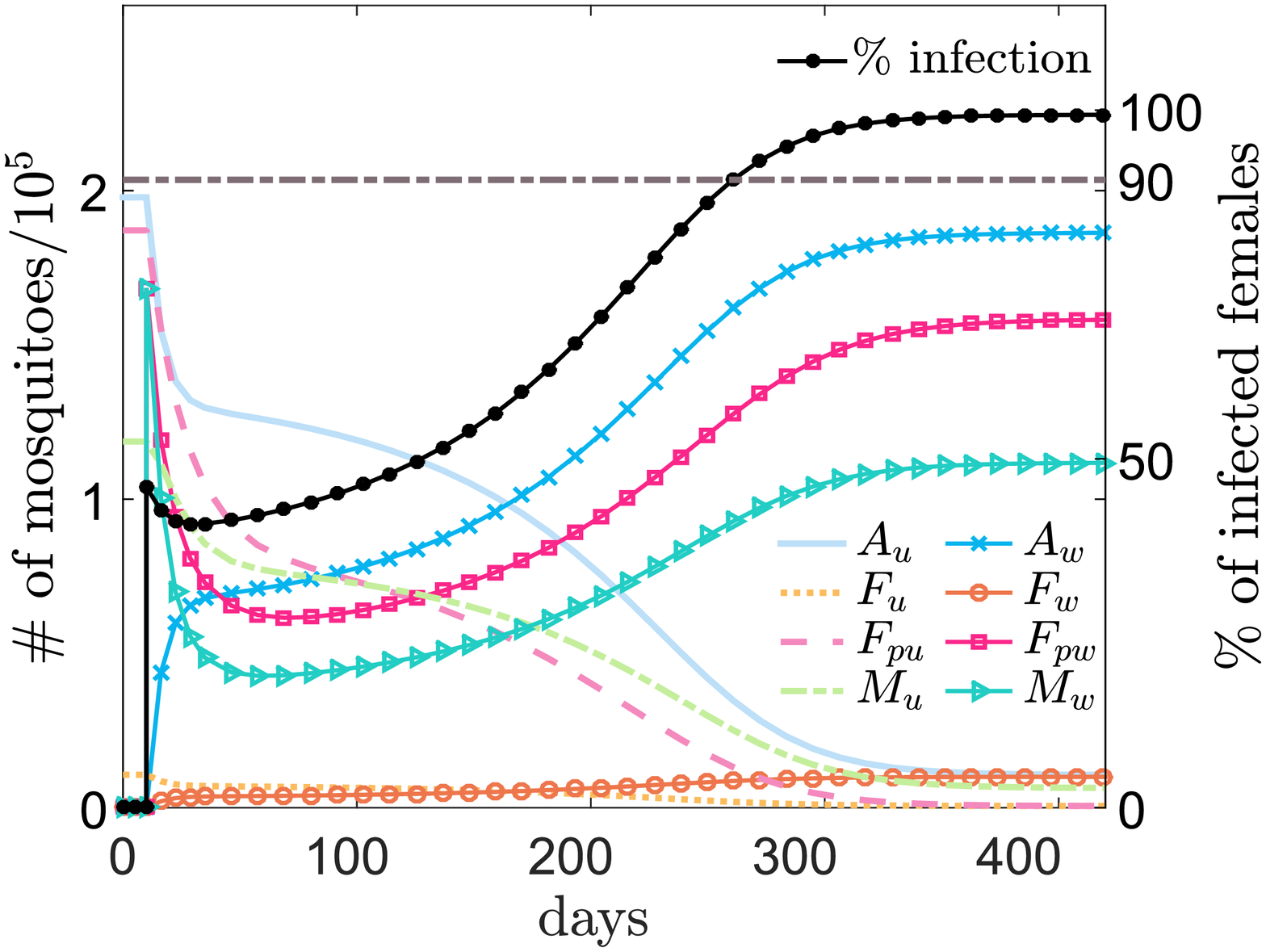}
\caption{PFR Approach Release infected pregnant females, $F_{pw}$, and males, $M_w$, at day 10 of the simulation. Left: In the first 20 days after the release, the infected adult populations all decrease, although the fraction of infected females (black circles) remains relatively constant. Right: A stable endemic state (90\% infection) is established just after 260 days.}\label{fig:ex2_release}
\end{figure}

\Cref{fig:ex2_release} shows the results when there is no pre-release mitigation, to reduce the DFE wild population, before 
releasing infected males and females (PFR approach) at day 10 of the simulation.  After a short initial transition stage, the infected population gradually dominates, and the a 90\% stable endemic infection  is achieved shortly after 260 days.  In similar simulation for the NPFR Approach, we see that there is a delay ($\sim$\,20 days, see row 1 in \cref{tab:ex2_90}) in the establishment of the epidemic when releasing nonpregnant females. This delay results from the time that it takes for a nonpregnant female to mate with a male and enter the pregnant stage.   

\begin{table}[tb]
\centering
\caption{Time (days) to 90\% infection in females. Release $X=0.9F^0_{pu}$ infected males and $X$ infected nonpregnant/pregnant females under different pre-release mitigation methods. Releasing pregnant females can establish an endemic state sooner than releasing nonpregnant females.  All the pre-release mitigation methods speed up the spreading of \textit{Wolbachia} infection.  The pre-release mitigation methods that target pregnant females (residual spraying, sticky trap) are more effective than ones that target only males or the aquatic-stage.   Residual spraying (in bold) is the most effective pre-release mitigation.\label{tab:ex2_90}}
\begin{tabular}{lrr}
\toprule
 Release Approach & $PFR (M_w+F_{pw})$ & NPFR ($M_w+F_{w})$ \\
\midrule
No pre-release mitigation &  261& 279 \\
\textbf{Residual spraying} & $\mathbf{52}$ & $\mathbf{55}$ \\
Larval control    & 203& 268 \\
Sticky trap      & 105 & 108 \\
Acoustic attraction& 215 & 227 \\ 
\bottomrule
\end{tabular}
\end{table}

\subsection{Q2: Which pre-release strategies are the most effective?}
Applying different pre-release mitigation methods (row 2-5 in \cref{tab:ex2_90}), we see that all the pre-release mitigation methods reduce the time to establish a \textit{Wolbachia}-infected population.  Ranking the pre-release strategies, in order of decreasing effectiveness, we observe: residual spraying $>$ sticky trap $>$ larval control $\approx$ acoustic attraction.
We also observe that it is less effective to release nonpregnant infected females (NPFR) rather than pregnant infected females (PFR). In practice, it is also more cost-effective to raise mosquitoes and store the new offspring (males and females) in the same container, we will only consider PFR Approach in the following simulations.

Our simulations (\cref{tab:ex2_90}) indicate that the pre-release mitigation methods targeting pregnant females (residual spraying, sticky trap) are more effective than ones that target only males or the aquatic-stage. In the maternal transmission cycle (\cref{fig:cycle}) at the DFE, most females are pregnant, $F^0_{pu}>>F^0_u$. By removing pregnant female mosquitoes before releasing the infected ones, it increases the fraction of new eggs that are infected, speeding up the spread of infection. Moreover, at the DFE there are many more uninfected males than nonpregnant females available (\cref{fig:ex2_release}), hence removing additional males does not greatly affect the overall transmission dynamics. 

By removing the aquatic-stage mosquitoes alone without killing any natural pregnant female is not an effective approach. When we remove the aquatic-stage mosquitoes, our aquatic populations are below the defined carrying capacity.  To exploit this gap we need to supply infected eggs as efficiently as we can. Since we don't directly release infected eggs into the environment and the majority of the pregnant females in the population are still uninfected, the gap made by the larval control is mostly filled with new uninfected eggs.

\subsection{Q3: Is One Big Release better than Split Repetitive Releases?}

Mitigation strategy Q3 considers if it is better to release all the infected mosquito at once or have repetitive smaller releases.
One advantage of the later case is that it may reduce the impact on local environment and neighborhood.

We simulate the practice of split releases by dividing one big release ($X$ males and $X$ pregnant females) into five smaller releases ($0.2X$ males and $0.2X$ pregnant females each time) with different releasing gaps (1 3, 7, 10 and 15 days between releases). We record the time of achieving 90\% infection in females for each possible combination of scenarios with different pre-release mitigation in \cref{tab:ex3_90_split}.

\begin{table}[tbhp]
\centering
\caption{The time (days) to achieve 90\% infection in females when releasing the same number of $X$ infected males and $X$ pregnant females in one big release or split repetitive releases with a time gap between the repetitive releases.   The pre-release mitigations are, from best to least effective:   Residual spraying $>$  Sticky trap $>$ Larval control $\approx$ Acoustic attraction. \label{tab:ex3_90_split}}
\begin{tabular}{lrrrrrrr}
\toprule
Time between releases:& Single & 1 day  & 3 days & 7 days & 10 days & 15 days \\
 & Release & gap & gap & gap & gap & gap \\
\midrule
No pre-release mitigation & 261 & 248 & 229 & 221 & 223 & 234 \\
\textbf{Residual spraying}& $\mathbf{52}$ & $\mathbf{54}$  & $\mathbf{66}$  & $\mathbf{80}$  & $\mathbf{90}$  & $\mathbf{103}$ \\
Larval control            & 203 & 229 & 214 & 208 & 210 & 221 \\
Sticky trap               & 105 & 108 & 118 & 131 & 142 & 159 \\
Acoustic attraction       & 215 & 207 & 199 & 200 & 207 & 222 \\ 
\bottomrule
\end{tabular}
\end{table}

Comparing the data in \cref{tab:ex3_90_split} horizontally in each row, we see that the optimal releasing interval may be different for different pre-release mitigation. For the case without pre-release mitigation, doing split releases is a better strategy than a big release. This is due to the constraint of carrying capacity in the aquatic-stage population. When releasing pregnant females without any pre-release mitigation (at DFE), the limited carrying capacity for producing infected offspring, limits the vertical transmission of the infection.  If the released infected males could sterilize a significant number of the uninfected nonpregnant females, then this can free up some of the aquatic carrying capacity and create space for the infected offspring to survive.  However, there are very few nonpregnant females,  and the male lifespan is short, one big release may result in a poor use of the males. Meanwhile, repetitive releases are able to solve those issues by reducing the redundancy in the males (smaller amount of infected males are released each time) and maintaining the availability of infected male over a longer time span (several releases with suitable time span in between).  Among these strategies, leaving a week between the releases is an optimal choice.

Similar explanations can be applied to understand the case of acoustic attraction. Even when 80\% of males are killed, the number of males is still much more than the nonpregnant females out there, and releasing all the infected males at once again leads to the waste of power of infected males. The optimal time period between releases is also around one week.

For pre-release mitigation using residual spraying or sticky trap, one big release is better than repetitive releases of smaller size. This is because a big change in the number of uninfected pregnant females may cause a gap in the aquatic-stage population, and it is critical to fill in the gap with infected offspring as soon as possible to avoid the bouncing back of natural population. Therefore, releasing all the infected mosquitoes at once is more efficient.

A similar case is for the pre-release mitigation using larval control, where it is better to release all the infected mosquitoes at the same time. However, split releases with just one day gap in between cause a big delay ($\sim$\,25 days) in the dynamics. This stresses the importance of bouncing back effect in the aquatic-stage population. Unlike the mitigation using residual spraying or sticky traps, larval control directly remove the aquatic-stage and create a gap. If only $0.2X$ infected pregnant females are introduced at the first release and there are $0.9X$ uninfected pregnant females in the field, then more than 80\% of the new offspring are uninfected, and they fill the gap immediately. As a result, the system arrives at a similar situation as it would be without any pre-release mitigation, and the corresponding trend in using a different releasing interval is similar to the first row.

In summary, to establish a \textit{Wolbachia}-infected population: 
\begin{itemize}
\item[A1:] Releasing infected males and pregnant females is more effective than releasing infected males and nonpregnant females;  
\item[A2:] Residual spraying, including the breeding sites, is a more effective pre-release strategy than the other pre-release strategies considered; and
\item[A3:] It better to release all the infected mosquitoes at once than to repetitively release smaller batches.
\end{itemize}

\section{Discussions and Conclusions}
We propose and analyze a two-sex multi-stage compartmental ODE model to describe the \textit{Wolbachia} transmission in a wild mosquito population. Our model captures the complex transmission cycle by including both male and female mosquitoes, nonpregnant and pregnant female mosquitoes, and aquatic-life stage mosquitoes limited by a prescribed carrying capacity. In particular, for any pregnant female mosquito, it can be in one of the three states: pregnant uninfected, pregnant sterile or pregnant infected. 

We study the existence and stability of the equilibrium points associated with the proposed model for both perfect and imperfect maternal transmissions. For both cases, there is one DFE, which is stable for $0<\R_0<1$. When the maternal transmission is perfect, there is a stable CIE and a unstable EE. Meanwhile, when the maternal transmission is imperfect, there is no CIE that could be achieved but two EE: a high-infection stable EE and a low-infection unstable EE. This stability analysis leads to a backward bifurcation diagram (see \cref{fig:bifur} and \cref{fig:bifur2}) with the unstable equilibrium points being the threshold condition for endemic \textit{Wolbachia}: below the threshold, the infection is wiped out by the wild uninfected mosquitoes and system goes back to DFE; above the threshold, the infection spreads out and eventually all/most mosquitoes are infected with \textit{Wolbachia}.

This threshold condition is described by three dimensionless numbers associated with the proposed model: basic reproductive number $\R_0$, the next generation number for the infected population $\G_{0w}$ and the next generation number for the uninfected population $\G_{0u}$, and we observe that $\R_0=\G_{0w}/\G_{0u}$.   \textit{Wolbachia}-infected populations of \textit{Ae. aegypti} are not found in nature, implying that  $\R_0<1$.  The basic reproductive number quantifies the local stability of the DFE when small numbers of infected mosquitoes are introduced.  There is a threshold condition where if a large number of infected mosquitoes are introduced, the population is attracted to a stable endemic \textit{Wolbachia}-infected equilibrium state.  We used backward bifurcation analysis to characterize this threshold condition to demonstrate that using \textit{Wolbachia} could be a  practical approach to suppress the spread of infectious diseases such as dengue fever, chikungunya and Zika.  

Our sensitivity analysis identified that the fitness cost (lifespan and egg laying rates) and the maternal transmission rate are two key factors to the potential success of creating endemic \textit{Wolbachia} in a wild mosquito population. A higher fitness cost or lower maternal transmission rate makes it the less efficient to establish the infection, which is reflected through higher threshold, longer spreading process and smaller final infection prevalence. The maternal transmission rate has the largest impact for all three aspects.

We found that releasing infected pregnant females was more effective than releasing infected nonpregnant females.  It is also more cost-efficient to raise mosquitoes and store infected males and females in the same container, resulting in releasing infected pregnant females. 

Our simulations indicate that the pre-release mitigations that target pregnant females, such as residual spraying and sticky gravid traps, are more helpful than ones that target only males or the aquatic-stage. Removing uninfected pregnant females greatly slows down the reproduction of the uninfected offspring, and the gap can be filled up mostly with infected population. Finally, we compare the efficiency between releasing all the infected mosquitoes at once and split releases of smaller sizes. Since repetitive releases are often done in the field, it is interesting to learn how this repetition and releasing interval will affect the disease transmission. The results show that, under different pre-release mitigations, different releasing strategy is desired, and it depends on what group of natural population that the pre-release mitigation targets.

This model offers important insights into using \textit{Wolbachia} as a potential mitigation strategy.  Before using these insights to guide policy, the uncertainty of the predictions must be quantified with respect to the model assumptions.  
For example, we have assumed that all the parameters are constant, thus there is no seasonal variation. In reality, parameters, such as development rate of the aquatic-stage, death rates, and the carrying capacity of the local environment, vary with the temperature and humidity. By including seasonality, the model would give more practical guide when the releasing process spans more than one season. Our model has assumed that all the mosquitoes are homogeneously mixed together, and the results can be considered as the average over a large number of random individual behavior. However, when it comes to the real field releases of \textit{Wolbachia}-infected mosquitoes, the infected population is only released at several distant spots, from where the infection diffuses out in a radial symmetric manner. We are currently extending this model to include both spatial heterogeneity and temporal variations using a partial differential equations that incorporate seasonal variations and the diffusion of mosquitoes.

\appendix
\begin{appendices}
\renewcommand{\theequation}{A.\arabic{equation}}
\setcounter{equation}{0}
\renewcommand{\thefigure}{A.\arabic{figure}}
\setcounter{figure}{0}
\renewcommand{\thesection}{A}
\setcounter{section}{0}

\subsection{Proof of \cref{thm:th_CIE}}\label{sec:A2}
At the CIE, the Jacobian is partitioned as
\begin{equation*}
J_{CIE}=\left(
\begin{array}{c;{2pt/2pt}c}
A_{CIE} & 0\\ \hdashline[2pt/2pt]
C_{CIE} & D_{CIE}
\end{array}
\right).
\end{equation*}

We first prove the stability of matrix
\begin{equation*}
A_{CIE} = \left(
\begin{array}{cccc}
-(\mu_a+\psi) & 0 & \frac{\phi_u}{\G_{0w}} & 0\\
b_f\psi & -\sigma-\mu_{fu} & 0 & 0\\
0 & 0 & -\mu_{fu} & 0\\
b_m\psi & 0 & 0 & -\mu_{mu}
\end{array}
\right).
\end{equation*}
The diagonal position $(4,4)$ of matrix $A_{CIE}$, $-\mu_{mu}<0$, is a negative eigenvalue, and therefore we need only consider the $3\times 3$ leading principal submatrix of $A_{CIE}$, which is partitioned as 
\begin{equation*}
A_{s2} =
\left(
\begin{array}{cc;{2pt/2pt}c}
-(\mu_a+\psi) & 0 & \frac{\phi_u}{\G_{0w}} \\
b_f\psi & -\sigma-\mu_{fu} & 0 \\ \hdashline[2pt/2pt]
0 & 0 & -\mu_{fu} 
\end{array}
\right)
 =\left(
\begin{array}{c;{2pt/2pt}c}
A_3 & B_3\\\hdashline[2pt/2pt]
C_3 & D_3
\end{array}
\right).
\end{equation*}
$A_{s2}$ is a Metzler matrix  \cite{kamgang2008computation} and is Metzler stable if and only if both $A_3$ and $D_3-C_3A_3^{-1}B_3=-\mu_{fu}<0$
are Metzler stable. The Metzler stability of $A_3$ is immediate, since it is a lower triangular matrix with negative diagonal entries and nonnegative off-diagonal entries. 

Now we consider the stability of
\begin{equation*}
D_{CIE} = \left(
\begin{array}{cccc}
-\G_{0w}(\mu_a+\psi) & 0 & \frac{\phi_w}{\G_{0w}} & 0\\
b_f\psi & -\sigma-\mu_{fw} & 0 & 0\\
0 & \sigma & -\mu_{fw} & 0\\
b_m\psi & 0 & 0 & -\mu_{mw}
\end{array}
\right).
\end{equation*}
The $(4,4)$ entry of $D_{CIE}$, $-\mu_{mw}<0$, is an eigenvalue and  therefore we need only consider the $3\times 3$ leading principal submatrix
\begin{equation*}
D_{s2} = \left(
\begin{array}{cc;{2pt/2pt}c}
-\G_{0w}(\mu_a+\psi) & 0 & \frac{\phi_w}{\G_{0w}} \\
b_f\psi & -\sigma-\mu_{fw} & 0 \\ \hdashline[2pt/2pt]
0 & \sigma & -\mu_{fw} 
\end{array}
\right) =\left(
\begin{array}{c;{2pt/2pt}c}
A_4 & B_4\\\hdashline[2pt/2pt]
C_4 & D_4
\end{array}
\right),
\end{equation*}
which is a Metzler matrix. The Metzler stability of matrices $A_4$ is obvious, and 
\begin{equation}
D_4-C_4A_4^{-1}B_4=-\mu_{fw}\left(1-\frac{1}{\G_{0w}}\right)<0 \quad\mbox{provided}\quad \G_{0w}>1.
\label{eq:pf2b}
\end{equation} 

Thus, under the condition in \cref{eq:pf2b}, the Metzler stability of Jacobian $J_{CIE}$ is guaranteed, that is, all the eigenvalues are negative and CIE is stable if $\G_{0w}>1$. The proof of theorem \cref{thm:th_CIE} is completed.

\end{appendices}

\section*{Acknowledgments} 
We are thankful to our colleagues Panpim Thongsripong and Dawn Wesson from School of Public Health and Tropical Medicine at Tulane University who provided expertise on \textit{Wolbachia} infection in mosquitoes. We thank Jeremy Dewar for his help with the local and global sensitivity analysis. This research was partially supported by the NSF/MPS/DMS-NIH/NIGMS award NSF-1563531 and the NIH-NIGMS Models of Infectious Disease Agent Study (MIDAS) award U01GM097661. The content is solely the responsibility of the authors and does not necessarily represent the official views of the National Science Foundation or the  National Institutes of Health.
\bibliographystyle{siamplain}
\bibliography{WolbachiaPregnancy}

\begin{thebibliography}{10}

\bibitem{achee2015critical}
{\sc N.~L. Achee, F.~Gould, T.~A. Perkins, R.~C. Reiner~Jr, A.~C. Morrison,
  S.~A. Ritchie, D.~J. Gubler, R.~Teyssou, and T.~W. Scott}, {\em A critical
  assessment of vector control for dengue prevention}, PLoS Negl Trop Dis, 9
  (2015), p.~e0003655.

\bibitem{alphey2010sterile}
{\sc L.~Alphey, M.~Benedict, R.~Bellini, G.~G. Clark, D.~A. Dame, M.~W.
  Service, and S.~L. Dobson}, {\em Sterile-insect methods for control of
  mosquito-borne diseases: an analysis}, Vector-Borne and Zoonotic Diseases, 10
  (2010), pp.~295--311.

\bibitem{AMCA}
{\sc {American Mosquito Control Association}}, {\em Control: Centers for
  disease control and prevention}.
\newblock \url{http://www.mosquito.org/control} (23 March 2017).

\bibitem{barrera2014use}
{\sc R.~Barrera, M.~Amador, V.~Acevedo, B.~Caban, G.~Felix, and A.~J. Mackay},
  {\em Use of the {CDC} autocidal gravid ovitrap to control and prevent
  outbreaks of {A}edes aegypti ({D}iptera: {C}ulicidae)}, Journal of medical
  entomology, 51 (2014), pp.~145--154.

\bibitem{belton1994attractton}
{\sc P.~Belton}, {\em Attractton of male mosquitoes to sound}, J. Am. Mosq.
  Control, 10 (1994), pp.~297--301.

\bibitem{CDCZika}
{\sc {Centers for Disease Control and Prevention}}, {\em Pregnancy: Centers for
  disease control and prevention}.
\newblock \url{https://www.cdc.gov/zika/pregnancy/index.html} (23 March 2017).

\bibitem{CDCDengueVaccine}
{\sc {Centers for Disease Control and Prevention}}, {\em Prevention: Centers
  for disease control and prevention}.
\newblock \url{https://www.cdc.gov/dengue/prevention/index.html} (23 March
  2017).

\bibitem{CDCChkVaccine}
{\sc {Centers for Disease Control and Prevention}}, {\em Prevention: Centers
  for disease control and prevention}.
\newblock \url{https://www.cdc.gov/chikungunya/prevention/index.html} (23 March
  2017).

\bibitem{CDCZikaVaccine}
{\sc {Centers for Disease Control and Prevention}}, {\em Treatment: Centers for
  disease control and prevention}.
\newblock \url{https://www.cdc.gov/zika/symptoms/treatment.html} (23 March
  2017).

\bibitem{chitnis2008determining}
{\sc N.~Chitnis, J.~M. Hyman, and J.~M. Cushing}, {\em Determining important
  parameters in the spread of malaria through the sensitivity analysis of a
  mathematical model}, Bulletin of mathematical biology, 70 (2008),
  pp.~1272--1296.

\bibitem{eliminatedengue}
{\sc {Eliminate Dengue Program}}, {\em Our research: {W}olbachia}.
\newblock \url{http://www.eliminatedengue.com/our-research/wolbachia} (29 March
  2017).

\bibitem{farkas2010structured}
{\sc J.~Z. Farkas and P.~Hinow}, {\em Structured and unstructured continuous
  models for {W}olbachia infections}, Bulletin of mathematical biology, 72
  (2010), pp.~2067--2088.

\bibitem{field1999microbe}
{\sc L.~Field, A.~James, M.~Turelli, and A.~Hoffmann}, {\em Microbe-induced
  cytoplasmic incompatibility as a mechanism for introducing transgenes into
  arthropod populations}, Insect molecular biology, 8 (1999), pp.~243--255.

\bibitem{foster2002mosquitoes}
{\sc W.~Foster and E.~Walker}, {\em Mosquitoes ({C}ulicidae)}, Medical and
  veterinary entomology,  (2002), pp.~203--262.

\bibitem{gubler1998dengue}
{\sc D.~J. Gubler}, {\em Dengue and dengue hemorrhagic fever}, Clinical
  microbiology reviews, 11 (1998), pp.~480--496.

\bibitem{hilgenboecker2008many}
{\sc K.~Hilgenboecker, P.~Hammerstein, P.~Schlattmann, A.~Telschow, and J.~H.
  Werren}, {\em How many species are infected with {W}olbachia? --a statistical
  analysis of current data}, FEMS microbiology letters, 281 (2008),
  pp.~215--220.

\bibitem{hoffmann2014stability}
{\sc A.~A. Hoffmann, I.~Iturbe-Ormaetxe, A.~G. Callahan, B.~L. Phillips,
  K.~Billington, J.~K. Axford, B.~Montgomery, A.~P. Turley, and S.~L. O'Neill},
  {\em Stability of the {wMel} {W}olbachia infection following invasion into
  {A}edes aegypti populations}, PLoS Negl Trop Dis, 8 (2014), p.~e3115.

\bibitem{hughes2013modelling}
{\sc H.~Hughes and N.~F. Britton}, {\em Modelling the use of {W}olbachia to
  control dengue fever transmission}, Bulletin of mathematical biology, 75
  (2013), pp.~796--818.

\bibitem{kamgang2008computation}
{\sc J.~C. Kamgang and G.~Sallet}, {\em Computation of threshold conditions for
  epidemiological models and global stability of the disease-free equilibrium
  ({DFE})}, Mathematical biosciences, 213 (2008), pp.~1--12.

\bibitem{keeling2003invasion}
{\sc M.~J. Keeling, F.~Jiggins, and J.~M. Read}, {\em The invasion and
  coexistence of competing {W}olbachia strains}, Heredity, 91 (2003),
  pp.~382--388.

\bibitem{koiller2014aedes}
{\sc J.~Koiller, M.~Da~Silva, M.~Souza, C.~Code{\c{c}}o, A.~Iggidr, and
  G.~Sallet}, {\em {A}edes, {W}olbachia and dengue}, PhD thesis, Inria
  Nancy-Grand Est (Villers-l{\`e}s-Nancy, France), 2014.

\bibitem{laven1956cytoplasmic}
{\sc H.~Laven}, {\em Cytoplasmic inheritance in {C}ulex}, Nature, 177 (1956),
  pp.~141--142.

\bibitem{lumjuan2005elevated}
{\sc N.~Lumjuan, L.~McCarroll, L.-a. Prapanthadara, J.~Hemingway, and
  H.~Ranson}, {\em Elevated activity of an epsilon class glutathione
  transferase confers {DDT} resistance in the dengue vector, {A}edes aegypti},
  Insect biochemistry and molecular biology, 35 (2005), pp.~861--871.

\bibitem{mcgraw2013beyond}
{\sc E.~A. McGraw and S.~L. O'neill}, {\em Beyond insecticides: new thinking on
  an ancient problem}, Nature Reviews Microbiology, 11 (2013), pp.~181--193.

\bibitem{mcmeniman2009stable}
{\sc C.~J. McMeniman, R.~V. Lane, B.~N. Cass, A.~W. Fong, M.~Sidhu, Y.-F. Wang,
  and S.~L. O'neill}, {\em Stable introduction of a life-shortening {W}olbachia
  infection into the mosquito {A}edes aegypti}, Science, 323 (2009),
  pp.~141--144.

\bibitem{mcmeniman2010virulent}
{\sc C.~J. McMeniman and S.~L. O'Neill}, {\em A virulent {W}olbachia infection
  decreases the viability of the dengue vector {A}edes aegypti during periods
  of embryonic quiescence}, PLoS Negl Trop Dis, 4 (2010), p.~e748.

\bibitem{ndii2015modelling}
{\sc M.~Z. Ndii, R.~I. Hickson, D.~Allingham, and G.~Mercer}, {\em Modelling
  the transmission dynamics of dengue in the presence of {W}olbachia},
  Mathematical biosciences, 262 (2015), pp.~157--166.

\bibitem{ndii2012modelling}
{\sc M.~Z. Ndii, R.~I. Hickson, and G.~N. Mercer}, {\em Modelling the
  introduction of {W}olbachia into {A}edes aegypti mosquitoes to reduce dengue
  transmission}, The ANZIAM Journal, 53 (2012), pp.~213--227.

\bibitem{nguyen2009evaluation}
{\sc H.~T. Nguyen, P.~I. Whelan, M.~S. Shortus, and S.~P. Jacups}, {\em
  Evaluation of bifenthrin applications in tires to prevent {A}edes mosquito
  breeding}, Journal of the American Mosquito Control Association, 25 (2009),
  pp.~74--82.

\bibitem{o2012open}
{\sc L.~O'Connor, C.~Plichart, A.~C. Sang, C.~L. Brelsfoard, H.~C. Bossin, and
  S.~L. Dobson}, {\em Open release of male mosquitoes infected with a
  {W}olbachia biopesticide: field performance and infection containment}, PLoS
  Negl Trop Dis, 6 (2012), p.~e1797.

\bibitem{pialoux2007chikungunya}
{\sc G.~Pialoux, B.-A. Ga{\"u}z{\`e}re, S.~Jaur{\'e}guiberry, and M.~Strobel},
  {\em Chikungunya, an epidemic arbovirosis}, The Lancet infectious diseases, 7
  (2007), pp.~319--327.

\bibitem{promsiri2006evaluations}
{\sc S.~Promsiri, A.~Naksathit, M.~Kruatrachue, and U.~Thavara}, {\em
  Evaluations of larvicidal activity of medicinal plant extracts to {A}edes
  aegypti ({D}iptera: {C}ulicidae) and other effects on a non target fish},
  Insect Science, 13 (2006), pp.~179--188.

\bibitem{ritchie2002dengue}
{\sc S.~A. Ritchie, J.~N. Hanna, S.~L. Hills, J.~P. Piispanen, W.~J.~H.
  McBride, A.~Pyke, and R.~L. Spark}, {\em Dengue control in north
  {Q}ueensland, {A}ustralia: case recognition and selective indoor residual
  spraying}, Dengue Bulletin, 26 (2002), pp.~7--13.

\bibitem{ritchie2009lethal}
{\sc S.~A. Ritchie, L.~Rapley, C.~Williams, P.~Johnson, M.~Larkman, R.~Silcock,
  S.~Long, and R.~Russell}, {\em A lethal ovitrap-based mass trapping scheme
  for dengue control in {A}ustralia: I. {P}ublic acceptability and performance
  of lethal ovitraps}, Medical and veterinary entomology, 23 (2009),
  pp.~295--302.

\bibitem{segoli2014effect}
{\sc M.~Segoli, A.~A. Hoffmann, J.~Lloyd, G.~J. Omodei, and S.~A. Ritchie},
  {\em The effect of virus-blocking {W}olbachia on male competitiveness of the
  dengue vector mosquito, {A}edes aegypti}, PLoS Negl Trop Dis, 8 (2014),
  p.~e3294.

\bibitem{styer2007mortality}
{\sc L.~M. Styer, S.~L. Minnick, A.~K. Sun, and T.~W. Scott}, {\em Mortality
  and reproductive dynamics of {A}edes aegypti ({D}iptera: {C}ulicidae) fed
  human blood}, Vector-borne and zoonotic diseases, 7 (2007), pp.~86--98.

\bibitem{tun2000effects}
{\sc W.~Tun-Lin, T.~Burkot, and B.~Kay}, {\em Effects of temperature and larval
  diet on development rates and survival of the dengue vector {A}edes aegypti
  in north {Q}ueensland, {A}ustralia}, Medical and veterinary entomology, 14
  (2000), pp.~31--37.

\bibitem{turelli1991rapid}
{\sc M.~Turelli and A.~A. Hoffmann}, {\em Rapid spread of an inherited
  incompatibility factor in {C}alifornia {D}rosophila}, Nature, 353 (1991),
  p.~440.

\bibitem{van2002reproduction}
{\sc P.~Van~den Driessche and J.~Watmough}, {\em Reproduction numbers and
  sub-threshold endemic equilibria for compartmental models of disease
  transmission}, Mathematical biosciences, 180 (2002), pp.~29--48.

\bibitem{walker2011wmel}
{\sc T.~Walker, P.~Johnson, L.~Moreira, I.~Iturbe-Ormaetxe, F.~Frentiu,
  C.~McMeniman, Y.~Leong, Y.~Dong, J.~Axford, P.~Kriesner, et~al.}, {\em The
  {wMel} {W}olbachia strain blocks dengue and invades caged {A}edes aegypti
  populations}, Nature, 476 (2011), pp.~450--453.

\bibitem{xue2016two}
{\sc L.~Xue, C.~A. Manore, P.~Thongsripong, and J.~M. Hyman}, {\em Two-sex
  mosquito model for the persistence of {W}olbachia}, Journal of Biological
  Dynamics,  (2016), pp.~1--22.

\end{thebibliography}

\end{document}